\documentclass[journal]{IEEEtran}
\usepackage{amsmath,graphicx}
\usepackage{pgf,tikz,pgfplots}

\usepackage{amsmath,amssymb,amsthm}
\usepackage{algorithmic,algorithm}
\usepackage{bm}
\usepackage{color}
\newtheorem{theorem}{Theorem}
\usepackage{caption}
\usepackage{comment}

\newtheorem{lemma}{Lemma}

\newtheorem{proposition}{Proposition}

\newtheorem{definition}{Definition}

\newcommand{\ba}{\mathbf{a}}
\newcommand{\bb}{\mathbf{b}}
\newcommand{\bc}{\mathbf{c}}
\newcommand{\be}{\mathbf{e}}
\newcommand{\bh}{\mathbf{h}}

\newcommand{\bu}{\mathbf{u}}
\newcommand{\bmu}{\boldsymbol{\mu}}
\newcommand{\bnu}{\boldsymbol{\nu}}
\newcommand{\bv}{\mathbf{v}}

\newcommand{\bw}{\mathbf{w}}
\newcommand{\bn}{\mathbf{n}}
\newcommand{\x}{\mathbf{x}}
\newcommand{\bY}{\mathbf{Y}}
\newcommand{\bN}{\mathbf{N}}
\newcommand{\bx}{\mathbf{x}}
\newcommand{\bX}{\mathbf{X}}

\newcommand{\bs}{\mathbf{s}}
\newcommand{\by}{\mathbf{y}}
\newcommand{\bz}{\mathbf{z}}
\newcommand{\bA}{\mathbf{A}}

\newcommand{\bS}{\mathbf{S}}
\newcommand{\bC}{\mathbf{C}}

\newcommand{\bH}{\mathbf{H}}

\newcommand{\I}{\mathcal{I}}
\newcommand{\RR}{\mathcal{R}}

\usepackage{subfigure}

\def\notvartualgraph{1}
\ifthenelse{\notvartualgraph=1}{

}
{

}

\begin{document}
\title{Quantized Constant-Envelope Waveform Design for Massive MIMO DFRC Systems}
\author{\IEEEauthorblockN{Zheyu Wu, Ya-Feng Liu, Wei-Kun Chen,  and Christos Masouros}
  	\thanks{Z. Wu and and Y.-F. Liu are with the State Key Laboratory of Scientific and Engineering Computing, Institute of Computational Mathematics and Scientific/Engineering Computing, Academy of Mathematics and Systems Science, Chinese Academy of Sciences, Beijing 100190, China (e-mail: \{wuzy, yafliu\}@lsec.cc.ac.cn). W.-K. Chen is with the School of Mathematics and Statistics/Beijing Key Laboratory on MCAACI, Beijing Institute of Technology, Beijing 100081, China (e-mail: chenweikun@bit.edu.cn). C. Masouros is with the Department of Electronic and Electrical Engineering, University College London, Torrington Place, London, WC1E 7JE, UK (e-mail: c.masouros@ucl.ac.uk).}
  }
\maketitle

\maketitle

\begin{abstract}
Both dual-functional radar-communication (DFRC) and massive multiple-input multiple-output (MIMO) have been recognized as  enabling technologies for 6G wireless networks. This paper considers the advanced waveform design for hardware-efficient massive MIMO DFRC systems. Specifically,  the transmit waveform is imposed with the quantized constant-envelope (QCE)  constraint, which facilitates the employment of low-resolution digital-to-analog converters (DACs) and power-efficient amplifiers. The waveform design problem is formulated as the minimization of  the mean square error (MSE) between the designed and desired beampatterns subject to the constructive interference (CI)-based communication quality of service (QoS) constraints and the QCE constraint. To solve the formulated problem, we first utilize the penalty technique to  transform the discrete problem into an equivalent continuous penalty model.   Then, we propose an inexact augmented Lagrangian method (ALM) algorithm for solving the penalty model. In particular, the ALM subproblem at each iteration is solved by a custom-built block successive upper-bound minimization (BSUM) algorithm, which admits closed-form updates, making the proposed inexact ALM algorithm computationally efficient. Simulation results demonstrate the superiority of the proposed approach over existing state-of-the-art ones.  In addition, extensive simulations are conducted to examine the impact of various system parameters on the trade-off between communication  and radar performances.

\end{abstract}
\vspace{-0.2cm}
\section{Introduction}
Communication-radar spectrum sharing (CRSS) has emerged as a promising paradigm to address the spectrum congestion issue in today's ever-expanding wireless networks \cite{CRSS}. A straightforward approach to achieve CRSS is to allow individual radar and communication systems to share the frequency band, which, however,  necessitates the exchange of side-information for managing cross-interference, leading to a prohibitive cooperation cost \cite{coexist}. As a step ahead from the coexistence scheme, the dual-functional radar-communication (DFRC) system integrates both radar and communication functionalities on a single platform, which not only eliminates the need for information exchange but also greatly reduces the hardware cost and system complexity. The DFRC system has  gained significant attention recently in both academia and industry. In particular, the DFRC system has been considered as an enabling technology for future IoT applications and integrated sensing and communications (ISAC) has been listed as one of the six major usage scenarios in future  6G wireless networks \cite{DFRCapp,DFRC2,DFRC3,DFRC4,survey_DFRC,survey_DFRC2,wang2023globally,ITU}. 

The performance of the DFRC system can be further enhanced when equipped with massive multiple-input multiple-output (MIMO) \cite{massivemimo2}, another key technology for 5G and 6G wireless networks, as the increased number of antennas offers greater flexibility to mitigate the multiuser interference and shape the spectral beampattern for radar applications. Despite the great potential of  massive MIMO DFRC systems, the high hardware cost and energy consumption pose  major practical challenges. In particular, the numbers of digital-to-analog converters (DACs) and power amplifiers (PAs),  which are the most power-hungry and costly components in  fully digital systems,  both scale with the number of transmit antennas. Therefore,   it is crucial to employ low-resolution DACs and power-efficient PAs in massive MIMO DFRC systems. It is well-known that when constant-envelope (CE) waveforms are transmitted,  the PAs can operate close to saturation without distortion and achieve their highest efficiency \cite{PAbook}.  
Quantized CE (QCE) transmission is a novel scheme that combines the CE requirement  on the transmit waveform  with the use of low-resolution DACs, in which case the transmit signals are constrained to have a fixed amplitude and their phases are limited to a finite set of values. The QCE transmission  facilitates the use of both low-resolution DACs and the most power-efficient PAs, making it a favorable transmission scheme in massive MIMO systems.   In this paper, we focus on the QCE transmit waveform design for massive MIMO DFRC systems. 
\subsection{Related Works}
Extensive research efforts have been devoted to the design and analysis of QCE precoding in massive MIMO communication-only systems \cite{SQUID,diversity,analysis2,trellis,GEMM,ciqce,ICASSP2023}.  The QCE precoding schemes can be broadly categorized into two classes: linear-quantized precoding and nonlinear precoding. Linear-quantized precoding schemes directly quantize the output of linear precoders, which, although simple,  often fail to provide satisfactory performance \cite{SQUID,diversity,analysis2}.  In contrast, nonlinear precoding schemes  optimize the transmit signal based on both the channel and the data symbols by solving appropriate optimization problems, which generally yield significantly superior performance compared to linear-quantized schemes \cite{trellis,GEMM,ciqce,ICASSP2023}. 

Various nonlinear precoders have been proposed based on different criteria. In particular, the concept of constructive interference (CI) \cite{CI1,CI2,CI3} has been incorporated into nonlinear  precoding design \cite{ciqce,ICASSP2023}. The main idea is to shape the multiuser interference (MUI) such that it aligns constructively with the signal of interest, thereby maximizing the useful signal power \cite{CItutorial}. Due to the ability to exploit MUI,  CI-based nonlinear precoders generally outperform those designed based on the   classical mean square error (MSE) metric, which treats   MUI as a detrimental factor. 
 In addition to numerous notable works on QCE precoding for communication systems, there are also a few works investigating  the QCE waveform design for massive MIMO radar systems; see \cite{radar_QCE_3,radar_QCE_1,radar_QCE_2} and the references therein. It is worth mentioning that as an extreme case of QCE transmission, one-bit waveform design (where one-bit DACs are employed) has attracted independent  research interests in both communication systems \cite{analysis1,CIfirst,CImodel,sep2,NL1P} and radar systems \cite{radar_1bit1,radar_1bit2,radar_1bit3}.


Note that, the aforementioned works on QCE waveform design focus on either communication-only or radar-only scenarios. To date, there have been very few studies on QCE waveform design for DFRC systems. As far as we know, the only existing works are \cite{DFRC1bit3,DFRC1bit2,DFRC1bit1},  all of which consider the extreme one-bit case. Specifically, in \cite{DFRC1bit3},  the one-bit DFRC waveform is designed to minimize the communication  MSE under the Cramér-Rao bound (CRB) constraint for radar direction-of-arrive estimation.  
 To encompass wider applications for radar (e.g., target tracking),  the works \cite{DFRC1bit2} and \cite{DFRC1bit1} adopt  the MSE between the designed and desired beampatterns and the similarity between the designed and desired waveforms as radar metrics, respectively, which are minimized (in weight) jointly with the communication MSE. 

There are two main limitations in the current studies on QCE waveform design for DFRC systems. 
First, the existing works \cite{DFRC1bit3,DFRC1bit2,DFRC1bit1} all adopt the MSE as the communication metric, which has been shown to be less effective   than the CI metric in the context of  nonlinear precoding. Although the works \cite{DFRC_CI1,DFRC_CI2,DFRC_CI3} have adopted the CI metric  in the waveform design for DFRC systems, they do not consider the QCE constraint. Second, as mentioned earlier, the existing works \cite{DFRC1bit3,DFRC1bit2,DFRC1bit1}  all focus on the one-bit case, and the proposed algorithms therein cannot be directly extended to handle the general QCE constraint. To the best of our knowledge, the general QCE waveform design for DFRC systems has not been explored in the literature, possibly due to the technical challenges introduced by both  the QCE constraint and the DFRC scenario.  
\vspace{-0.2cm}
 \subsection{Our Contributions}
In this paper, we consider the general QCE waveform design for massive MIMO DFRC systems. The main contributions are summarized as follows.

\begin{itemize}
\item 
\emph{Novel CI-based QoS-constrained general QCE waveform design problem formulation.}  For the first time, we investigate the  waveform design problem for DFRC systems under general QCE constraints. We employ CI as the communication metric and formulate the problem as the minimization of the MSE between the designed and desired beampatterns subject to the CI-based communication QoS constraints and the general QCE constraint.  
This results in a large-scale optimization problem with a nonconvex quartic objective function, numerous linear constraints, and discrete QCE constraints.
\item \emph{Efficient approach for solving the formulated problem.} We propose an efficient approach for solving the formulated problem. Specifically, we first apply the penalty technique to transform the discrete model with complicated QCE constraints into a continuous penalty model, and theoretically establish the equivalence between these two models. Then, we propose an inexact augmented Lagrangian method (ALM) algorithm for solving the penalty model and prove its convergence. In particular,  we  propose a custom-built successive upper-bound minimization (BSUM) algorithm for solving the ALM subproblem at each iteration. The proposed BSUM algorithm admits closed-form updates, making the inexact ALM algorithm computationally efficient and suitable for solving large-scale problems. 
\end{itemize}

 We provide extensive simulation results to verify the effectiveness of the proposed approach. For the one-bit case, the proposed approach demonstrates significantly superior performance compared to existing state-of-the-art (SOTA) algorithms \cite{DFRC2,DFRC1bit1}. In particular, with the same communication symbol error rate (SER), the  beampattern MSE achieved by the proposed approach is generally more than three times lower than those achieved by the SOTA algorithms.  We also examine the impact of various system parameters, including the numbers of transmit antennas and users in the system, the quantization level, and the safety margin threshold involved in the CI constraint, on the trade-off between communication and radar performances.  In particular, we demonstrate that the shape and power level of the achieved beampattern are strongly influenced by the number of transmit antennas.  Additionally, we show that increasing the resolution of DACs from one bit to three bits can greatly enhance the radar and communication performance, but further increases beyond three bits yield only marginal gains. This highlights the advantages of equipping massive MIMO  and low-resolution DACs in DFRC systems.

 \vspace{-0.3cm}
\subsection{Organization and Notations}
The rest of the paper is organized as follows. In Section \ref{sec:2}, we introduce the system model and formulate the problem. In Section \ref{sec:3}, we propose an efficient approach for solving the formulated problem. Extensive simulation results are presented in Section \ref{sec:4} and the paper is concluded in Section \ref{sec:5}.

 \emph{Notation:} We use $x$, $\x$, $\mathbf{X}$, and $\mathcal{X}$ to denote scalar, column vector, matrix, and set, respectively.  For a matrix $\mathbf{X}$, we denote $\x_n$ as its $n$-th column and $x_{n,t}$ as its $(n,t)$-th entry. For a set $\mathcal{X}$, $\text{conv}(\mathcal{X})$ is the convex hull of $\mathcal{X}$ and  $\mathcal{P}_\mathcal{X}(\cdot)$ is the projection operator onto $\mathcal{X}$ (if $\mathcal{X}$ is convex). 
 $\|\cdot\|$ denotes the $\ell_2$ norm of the corresponding  vector or the spectrum norm of the corresponding matrix. The operators $(\cdot)^\mathsf{T}$, $(\cdot)^\mathsf{H}$, $\RR(\cdot)$, $\I(\cdot)$, and $|\cdot|$ return the transpose, the Hermitian transpose, the real part, the imaginary part, and the modular of their corresponding argument, respectively. The notation $\left<\x,\mathbf{y}\right>$ is also used in some contexts to denote the inner product of vectors $\mathbf{x}$ and $\mathbf{y}$ for clarity, as an alternative to $\mathbf{x}^\mathsf{T}\mathbf{y}$. The symbol $\mathbf{0}$  represents the all-zero vector with appropriate dimension.  $\mathcal{C}\mathcal{N}(\mathbf{0},\sigma^2\mathbf{I})$ represents the zero-mean circularly symmetric complex Gaussian distribution with covariance matrix $\sigma^2\mathbf{I}$, where $\mathbf{I}$ denotes the identity matrix.  For a positive integer $n$, we use $[n]$ to denote $\{1,2, \dots, n\}$. Finally, $\otimes$ represents the Kronecker product and $j\triangleq\sqrt{-1}$ is the imaginary unit.

\section{System Model and Problem Formulation}\label{sec:2}
\begin{figure}
\includegraphics[scale=0.25]{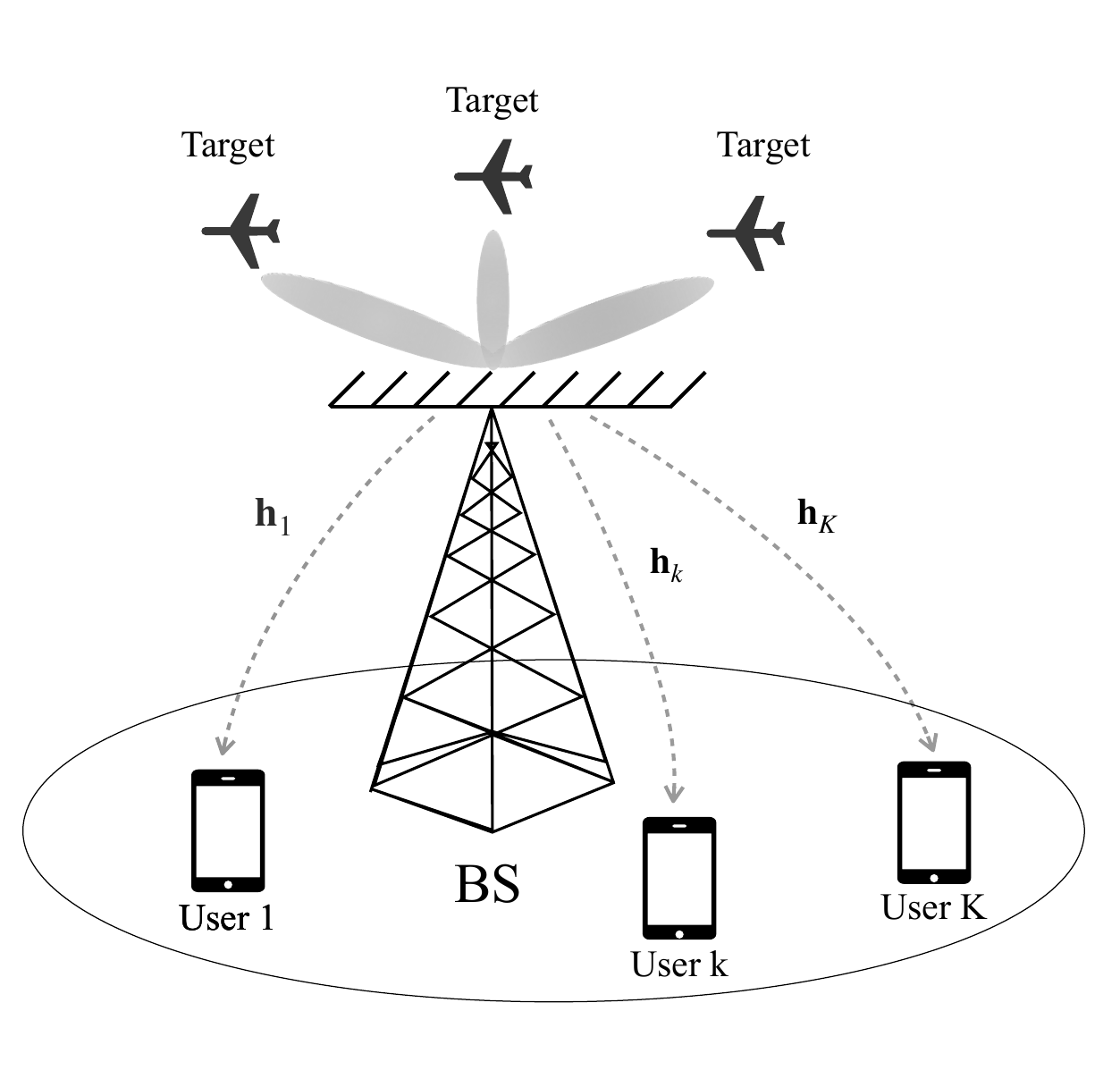}
\centering
\caption{A MIMO DFRC system, where there are one BS equipped with $N$ transmit antennas, each employed with a pair of low-resolution DACs, $K$ communication users, and many targets.}
\label{fig_system}
\vspace{-0.35cm}
\end{figure}

As depicted in Fig. \ref{fig_system}, we consider a MIMO DFRC system that simultaneously sends communication symbols to $K$ single-antenna users and transmits probing waveform to the targets for sensing purposes.  The system is equipped with a uniform linear array (ULA) with $N$ transmit antennas, each employed with a pair of low-resolution DACs to enhance the hardware efficiency. 
 
Let $\bX=[\x_1,\x_2,\dots,\x_T]\in\mathbb{C}^{N\times T}$ be the transmit signal  matrix, where $T$ is the block length and $\x_t$ is the transmit signal vector at the $t$-th  time slot.  
Typically, the transmit signal matrix should satisfy the CE constraint to achieve the highest power efficiency of PAs \cite{DFRC2,DFRC_CI1}, i.e., 
$$|x_{t,n}|=\eta,~\forall~n\in[N],~\forall~t\in[T],$$
where $\eta=\sqrt{{P}/{N}}$ with  $P$ being the maximum available power at each time slot.  As low-resolution DACs are employed,  each transmit signal is further restricted to be selected from only a finite set, i.e., 
$$x_{t,n}\in\mathcal{X}_L:=\left\{\eta e^{j\frac{(2\ell-1)\pi}{L}},~\ell\in[L]\right\},~\forall~n\in[N],~\forall~t\in[T].$$ 
By convention, we refer to $\mathcal{X}_L$  as the QCE set (as it is a quantized version of the CE signal set)  
and  $L$ as the number of quantization levels. Clearly, an $L$-level quantization can be realized by a pair of $(\log_2 L-1)$-bit DACs \cite{analysis2}. In this paper,  the transmit signal matrix $\bX$ is designed to  optimize the radar performance with a guaranteed communication QoS requirement. Next, we will introduce the models and performance metrics for communication and radar sensing, respectively.
\vspace{-0.2cm}
\subsection{Communication Model and Performance Metric}
The received signal matrix $\bY=[\by_1,\by_2,\dots,\by_T]\in\mathbb{C}^{K\times T}$ at the communication users is given by 
\begin{equation}\label{model:commun}
\bY=\bH\bX+\bN,
\end{equation}
where  $\bH=[\bh_1,\bh_2,\dots,\bh_K]^\mathsf{T}\in\mathbb{C}^{K\times N}$ is the channel matrix between the transmitter and the users, and $\bN=[\bn_1,\bn_2,\dots,\bn_T]\in\mathbb{C}^{K\times T}$ is the additive white Gaussian noise matrix with $\bn_t\sim\mathcal{CN}(\mathbf{0},\sigma^2\mathbf{I})$ for all $t\in[T]$. 
Let $\bS=[\bs_1,\bs_2,\dots,\bs_T]\in\mathbb{C}^{K\times T}$ be the data symbol matrix for the users.  The transmit signal matrix $\bX$ is designed based on both the channel $\bH$ and the data symbol matrix $\bS$ through nonlinear precoding. This kind of  precoding scheme can achieve significantly better performance than conventional linear precoding schemes designed solely based  on $\bH$, especially when low-resolution DACs are employed \cite{SQUID}.

In this paper, we adopt the CI metric as the communication performance metric and focus on the $M$-ary phase shift keying (PSK) constellation\footnote{We focus on PSK constellation for ease of presentation, but the results in this paper can be straightforwardly generalized to other constellation schemes, e.g., quadrature amplitude modulation (QAM).}.  CI refers to the interference that aligns constructively with the signal of interest, thereby enhancing the useful signal power and pushing the signal away from its corresponding decision boundary \cite{CI3}. Exploiting the idea of CI, the CI metric aims to maximize the distance between the received noise-free signal and its closest decision boundary of the data symbol, a crucial quantity known as the safety margin \cite{CImodel}.  The safety margin is closely related to the symbol error probability (SEP). Intuitively, with a larger safety margin, the received noise-free  signal would be less susceptible to the additive noise, resulting in a lower SEP.  More specifically,  it has been shown in \cite{diversity} that the SEP can be both lower and upper bounded by a function of the safety margin:  
\begin{equation}\label{upperbound}
Q\left(\frac{\sqrt{2}d}{\sigma}\right)\leq\text{SEP}\leq 2Q\left(\frac{\sqrt{2}d}{\sigma}\right),
\end{equation}
where $Q(x)=\frac{1}{\sqrt{2\pi}}\int_{x}^{\infty}e^{-\frac{1}{2}x^2}dx$ , $d$ is the safety margin, and $\sigma$ is the standard variance of the additive noise. Therefore, the CI-based communication QoS constraint can be formulated  as the form $d\geq b,$ where $b$ is the lower bound of the safety margin determined by the prespecified SEP threshold.

\begin{figure}
\centering
\includegraphics[scale=0.3]{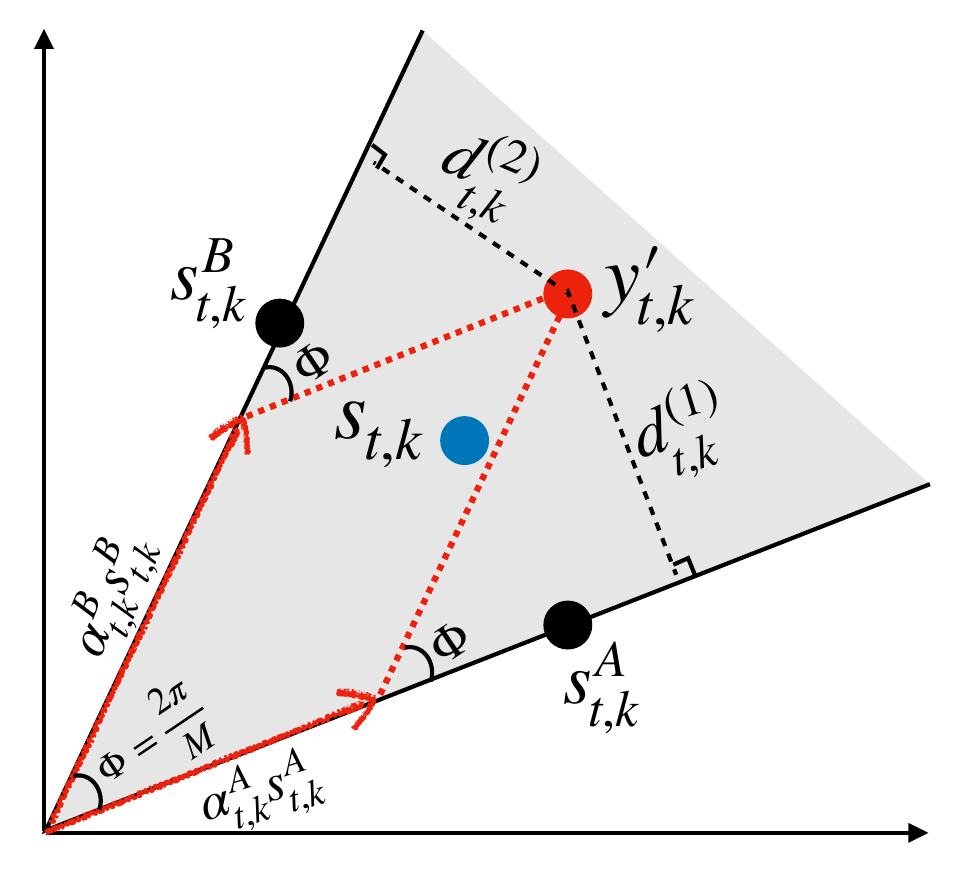}
\caption{An illustration of the CI metric.}
\label{fig_CI}
\end{figure}
Next,  we derive the CI-based communication QoS constraint for  model \eqref{model:commun}.  For illustration,  we depict a piece of the decision region of $8$-PSK constellation in Fig. \ref{fig_CI}, where $s_{t,k}$ and $y_{t,k}':=\bh_k^\mathsf{T}\x_t$ denote the data symbol and the noise-free received signal of user $k$ at the $t$-th time slot, respectively.  Clearly, the safety margin, i.e., the distance from $y_{t,k}'$ to its closest decision boundary of $s_{t,k}$, is given by $\min\left\{d_{t,k}^{(1)}, d_{t,k}^{(2)}\right\}$. Therefore, the block-level CI-based communication QoS constraint can be  formulated as \begin{equation}\label{eqn:QoS}\min\{d_{t,k}^{(1)}, d_{t,k}^{(2)}\}\geq b_{t,k},~\forall~ k\in[K],~\forall~t\in[T],\end{equation}
 where $b_{t,k}>0$ is the preset safety margin threshold for the $k$-th user at the $t$-th time slot. 
 
We still need to derive the explicit formula of $\min\{d_{t,k}^{(1)}, d_{t,k}^{(2)}\}$. Following \cite{CImodel}, we decompose $y_{t,k}'$ along the decision boundaries as 
\begin{equation}\label{eqn:decompose}
y_{t,k}'=\alpha_{t,k}^As_{t,k}^A+\alpha_{t,k}^Bs_{t,k}^B,
\end{equation}  where  $s_{t,k}^A=s_{t,k}e^{-j\frac{\pi}{M}}$ and $s_{t,k}^B=s_{t,k}e^{j\frac{\pi}{M}}$ are the unit vectors on the decision boundaries, as shown in Fig. \ref{fig_CI}. Then, we have
  $$\min\{d_{t,k}^{(1)}, d_{t,k}^{(2)}\}=\min\{\alpha_{t,k}^A\sin\Phi,\alpha_{t,k}^B\sin\Phi\},$$ where $\Phi=\frac{2\pi}{M}$. 
  By further rewriting \eqref{eqn:decompose} into the real space and noting that $y_{t,k}'=\bh_k^\mathsf{T}\x_t$, we can express $\alpha_{t,k}^A$ and $\alpha_{t,k}^B$ as 
\begin{subequations}
\small\begin{align}
\hspace{-0.1cm}\left(\begin{matrix}\alpha_{t,k}^A\\\alpha_{t,k}^B\end{matrix}\right)&\hspace{-0.05cm}=\hspace{-0.05cm}\frac{1}{\sin\Phi}\hspace{-0.05cm}\left(\begin{matrix}\hspace{-0.05cm}
\I(s_{t,k}^B)&\hspace{-0.25cm}-\RR(s_{t,k}^B)\\\hspace{-0.07cm}-\I(s_{t,k}^A)&\hspace{-0.15cm}\RR(s_{t,k}^A)\end{matrix}\right)\hspace{-0.1cm}\left(\begin{matrix}\RR(\bh_k^\mathsf{T})&\hspace{-0.15cm}-\I(\bh_k^\mathsf{T})\\\I(\bh_k^\mathsf{T})&\hspace{-0.1cm}\RR(\bh_k^\mathsf{T})\end{matrix}\right)\hspace{-0.15cm}\left(\begin{matrix}\RR(\x_t)\\\I(\x_t)\end{matrix}\right)\notag\\
&\hspace{-0.1cm}:=\hspace{-0.05cm}\frac{1}{\sin\Phi}\hspace{-0.05cm}\left(\begin{matrix}\bc_{t,2k-1}^\mathsf{T}\vspace{0.05cm}\\\bc_{t,2k}^\mathsf{T}\end{matrix}\right)\left(\begin{matrix}\RR(\x_t)\vspace{0.05cm}\\\I(\x_t)\end{matrix}\right), \tag{5}\label{CI}
\end{align}
\end{subequations}
where we have used the fact that $\RR(s_{t,k}^A)\mathcal{I}(s_{t,k}^B)-\RR(s_{t,k}^B)\mathcal{I}(s_{t,k}^A)=\sin(\arg(s_{t,k}^B)-\arg(s_{t,k}^A))=\sin\Phi$; see \cite{CImodel} for a detailed derivation. 
Let $\x_t^\RR=[\RR(\x_t)^\mathsf{T}~~\I(\x_t)^\mathsf{T}]^\mathsf{T}\in\mathbb{R}^{2N}$ be the real space representation of $\x_t$, $\bC_t=[\bc_{t,1},\bc_{t,2},\dots,\bc_{t,2K-1},\bc_{t,2K}]^\mathsf{T}\in\mathbb{R}^{2K\times2N}$ with $\bc_{t,k}$ defined in \eqref{CI}, and $\bb_t=[b_{t,1},b_{t,2},\dots,b_{t,K}]^\mathsf{T}\otimes[1,1]^\mathsf{T}\in\mathbb{R}^{2K}$. Then the constraint in \eqref{eqn:QoS} can be expressed as\vspace{-0.05cm}
$$\bC_t\x_t^\RR\geq \bb_t,~\forall~t\in[T].$$
\vspace{-0.6cm}
\subsection{Radar Model and Performance Metric }
Under the assumption that the transmitted probing waveform $\bX$ is narrow-band and that the propagation is line of sight (LoS), the baseband signal at angular $\theta\in[-\frac{\pi}{2},\frac{\pi}{2})$ at the $t$-th time slot is  \vspace{-0.1cm}$$r_t(\theta)=\ba(\theta)^\mathsf{H}\x_t,$$ where $\ba(\theta):=[1, e^{\pi\sin\theta},\dots,e^{j\pi(N-1)\sin\theta}]^\mathsf{T}$ is the steering vector and the antenna spacing is assumed to be half of the wavelength \cite{radaronly}. The beampattern,  which describes the power distribution of the probing waveform in the spacial domain, is given by\vspace{-0.2cm}
$$P(\theta)=\frac{1}{T}\sum_{t=1}^T|r_t(\theta)|^2=\frac{1}{T}\ba(\theta)^\mathsf{H}\bX\bX^\mathsf{H}\ba(\theta).\vspace{-0.1cm}$$

For radar applications such as detection and tracking, the probing waveform should exhibit strong power in the directions of potential targets while keeping power levels minimal  in other directions. To achieve this goal, a common design criterion is to match the designed beampattern $P(\theta)$ with a desired beampattern $d(\theta)$. Specifically, we adopt the MSE between the designed and desired beampatterns as the radar metric: \vspace{-0.15cm} 
 \begin{equation*}\label{eqn:sensing}
\frac{1}{Q}\sum_{q=1}^Q\left|\alpha d(\theta_q)-\frac{1}{T}\ba(\theta_q)^\mathsf{H}\bX\bX^\mathsf{H}\ba(\theta_q)\right|^2,\vspace{-0.2cm}
\end{equation*}
 where $\alpha>0$ is a scaling factor and $\{\theta_q\}_{q=1}^Q$ are the sampled angle grids.
 \vspace{-0.4cm}
\subsection{Problem Formulation}
Based on the above discussions, the QCE transmit waveform design problem, aiming to minimize the MSE between the designed and desired beampatterns under the CI-based communication QoS constraint, can be formulated as follows:
\begin{equation}\label{eqn:problem}
\begin{aligned}
\min_{\alpha>0, \bX}~&\frac{1}{Q}\sum_{q=1}^Q\left|\alpha d(\theta_q)-\frac{1}{T}\ba(\theta_q)^\mathsf{H}\bX\bX^\mathsf{H}\ba(\theta_q)\right|^2\\
\text{s.t. }~~&\bC_t\x_t^\RR\geq \bb_t,~\forall~t\in[T],\\
~&x_{t,n}\in\mathcal{X}_L, ~\forall ~n\in[N],~\forall~t\in[T], 
\end{aligned}
\end{equation}
where we recall that $\bX=[\x_1,\x_2,\dots,\x_T]$ and $\x_t^\RR$ is the real-space representation of $\x_t$.  To solve \eqref{eqn:problem}, we first note that for a given $\bX$, the problem is quadratic in $\alpha$, which admits a closed-form solution as 
$$\alpha^*(\bX)=\frac{\frac{1}{T}\sum_{q=1}^Qd(\theta_q)\ba(\theta_q)^\mathsf{H}\bX\bX^\mathsf{H}\ba(\theta_q)}{\sum_{q=1}^Qd^2(\theta_q)}.$$
Substituting $\alpha^*(\bX)$  into the objective function, problem \eqref{eqn:problem} is transformed into the following equivalent form (where we have omitted a constant scaling factor $\frac{1}{QT^2}$ in the objective function):
 \begin{equation}\label{eqn:problem2}
\begin{aligned}
\min_{\{\x_t\}}~&\sum_{q=1}^Q\left(\sum_{t=1}^T|\ba(\theta_q)^\mathsf{H}\x_t|^2\right)^2-\left(\sum_{q=1}^Q\sum_{t=1}^Tc_q|\ba(\theta_q)^\mathsf{H}\x_t|^2\right)^2\\\
\text{s.t. }~~&\bC_t\x_t^\RR\geq \bb_t,~\forall~t\in[T],\\
~&x_{t,n}\in\mathcal{X}_L, ~\forall ~n\in[N],~\forall~t\in[T],
\end{aligned}
\end{equation}
where  \begin{equation}\label{def:cq}
c_q=\frac{d(\theta_q)}{\sqrt{\sum_{q=1}^Q d^2(\theta_q)}},~\forall~q\in[Q].
\end{equation} By further rewriting \eqref{eqn:problem2} into the real space, we obtain 
\begin{subequations}\label{eqn:problem3}
\begin{align}
\hspace{-0.24cm}\min_{\{\x_t^\RR\}}\,~&\sum_{q=1}^Q\hspace{-0.05cm}\left(\sum_{t=1}^T\|\bA_q\x_t^\RR\|^2\right)^2\hspace{-0.15cm}-\hspace{-0.05cm}\left(\sum_{q=1}^Q\sum_{t=1}^Tc_q\|\bA_q\x_t^\RR\|^2\right)^2\label{obj}\\
\text{s.t. }~~&\bC_t\x_t^\RR\geq \bb_t,~\forall ~t\in[T],\label{linear}\\
&\x_{t,n}^\RR\in\mathcal{X}_L^\RR,~\forall~ n\in[N],~\forall~t\in[T],\label{qce}
\end{align}
\end{subequations}
where in \eqref{obj},
$$\begin{aligned}
\bA_q=\left(\begin{matrix}\RR(\ba(\theta_q))&\I(\ba(\theta_q))\\-\I(\ba(\theta_q))&\RR(\ba(\theta_q))\end{matrix}\right)\in\mathbb{R}^{2\times2N}\end{aligned},$$ 
 and in \eqref{qce}, $$\x_{t,n}^\RR=\left[\begin{matrix}x^\RR_{t,n}\\x^\RR_{t,n+N}\end{matrix}\right]\in\mathbb{R}^2, ~\mathcal{X}_L^\RR=\left\{\eta\left[\begin{matrix}\cos\frac{(2\ell-1)\pi}{L}\\\sin\frac{(2\ell-1)\pi}{L}\end{matrix}\right],~\ell\in[L]\right\}.$$  
 
In the rest of this paper, we will focus on problem \eqref{eqn:problem3}. From an optimization viewpoint, there are significant technical difficulties and challenges in solving problem \eqref{eqn:problem3}. First, the discrete QCE constraint \eqref{qce} is complicated, which poses a major difficulty. Second, the linear constraints in \eqref{linear} further complicates the solution of the problem. 
As shown in \cite{NL1P}, it is even NP-hard to check whether the problem is feasible with constraints \eqref{linear} and \eqref{qce}. Moreover, the objective function \eqref{obj} is nonconvex and has a complicated expression. 
 Finally, the problem size is large, especially when massive MIMO DFRC systems are considered.  In the next section, we shall propose an efficient approach for solving \eqref{eqn:problem3} by judiciously exploiting the problem structure.

\section{Proposed Approach}\label{sec:3}
In this section, we propose an efficient approach for solving problem \eqref{eqn:problem3}. First, we employ the penalty technique to tackle the discrete QCE constraint in \eqref{qce}, transforming the discrete problem into a continuous penalty model in   Section \ref{sec:penalty}. Then, we propose an inexact ALM algorithm for solving the penalty model and establish its convergence  in Section \ref{sec:alm}. In particular, we propose a custom-designed BSUM
algorithm for solving the ALM subproblem in Section \ref{sec:BSUM}. The updates of all blocks of variables in the proposed BSUM admit closed-from solutions, making the proposed inexact ALM algorithm computationally very efficient.  
\subsection{Penalty Model}\label{sec:penalty}
\begin{figure}
\centering
\includegraphics[scale=0.16]{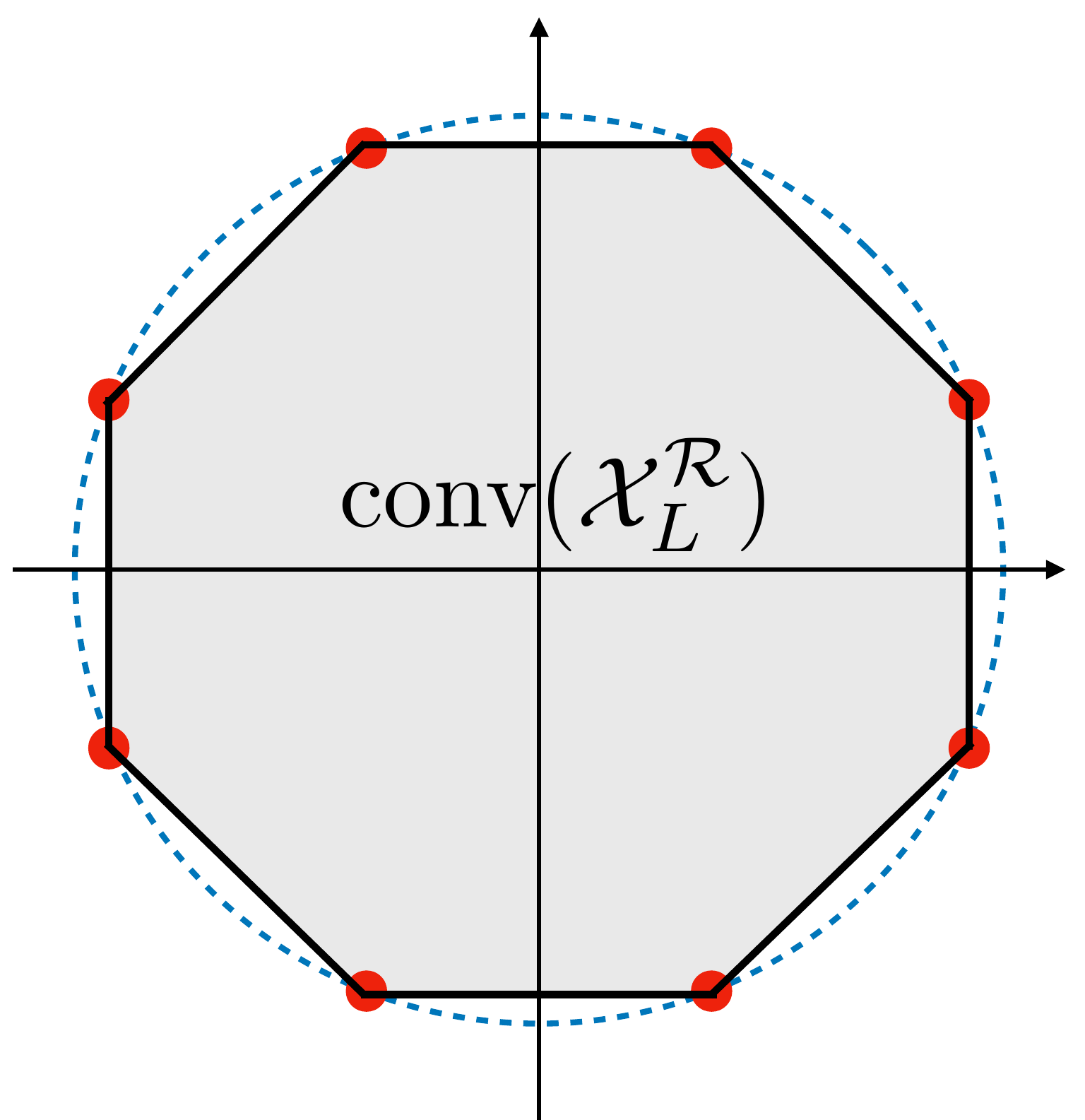}
\caption{An illustration of $\mathcal{X}_L^\RR$ and $\text{conv}(\mathcal{X}_L^\RR)$ with $L=8,$ where $\mathcal{X}_L^\RR$ is the set of 8 red QCE points and its convex hull $\text{conv}(\mathcal{X}_L^\RR)$ is the shadow region.}
\vspace{-0.2cm}
\label{convhull}
\end{figure}

The penalty technique is a powerful tool to deal with complicated constraints involved in optimization problems \cite{JSACsurvey}. Here, we use the penalty technique to handle   the discrete QCE constraint. 
Specifically, we relax the QCE constraint into its convex hull and include a negative square penalty term into the objective function as follows: 
\begin{equation}\label{eqn:problem4}
\begin{aligned}
\min_{\{\x_t^\RR\}}\,~&\sum_{q=1}^Q\left(\sum_{t=1}^T\|\bA_q\x_t^\RR\|^2\right)^2\\&-\left(\sum_{q=1}^Q\sum_{t=1}^Tc_q\|\bA_q\x_t^\RR\|^2\right)^2-\lambda\sum_{t=1}^T\|\x_t^\RR\|^2\\
\text{s.t. }~\,&\bC_t\x_t^\RR\geq \bb_t,~\forall~t\in[T],\\
&\x_{t,n}^\RR\in\text{conv}(\mathcal{X}_L^\RR), ~\forall~n\in[N],~\forall~t\in[T],\end{aligned}
\end{equation}
where $\lambda$ is the penalty parameter. The idea is based on the observation that the discrete QCE points in $\mathcal{X}_L^\RR$ are the vertices  of $\text{conv}(\mathcal{X}_L^\RR)$, which are the points in $\text{conv}(\mathcal{X}_L^\RR)$ with  the largest $\ell_2$ norm; see Fig. \ref{convhull}. Hence, by promoting a larger $\ell_2$ norm, the negative square penalty term, i.e., $-\lambda\sum_{t=1}^T\|\x_t^\RR\|^2$, encourages the solution of problem \eqref{eqn:problem4} to approach a discrete QCE point as the penalty parameter $\lambda$ increases. In the following proposition, we theoretically establish the equivalence between the continuous penalty model \eqref{eqn:problem4} and the original discrete model \eqref{eqn:problem3}.
\begin{proposition}\label{the1}
There exists $\lambda_0>0$ such that for all $\lambda>\lambda_0$, problems \eqref{eqn:problem3} and \eqref{eqn:problem4} share the same optimal solutions.
\end{proposition}
\begin{proof}
See Appendix \ref{appendixA}.
\end{proof}
In practice, it is difficult to determine the explicit value of $\lambda_0$ in Proposition \ref{the1}. In our implementation, we employ a common technique in penalty-based approaches known as the homotopy technique \cite{nonlinear}, which initializes the penalty parameter $\lambda$ with a small value and gradually increases it, tracking the solution path of the corresponding penalty models. Through this process, we will eventually find the required $\lambda$ in Proposition \ref{the1}. Moreover, this technique can greatly enhance the numerical performance of the penalty approaches \cite{GEMM,NL1P}, compared to directly solving the penalty model \eqref{eqn:problem4}  with a very large $\lambda$.

 We remark here that similar penalty strategies have been employed to deal with  the QCE constraint in communication-only QCE precoding design  \cite{GEMM,ICASSP2023}. However, the optimization problems in \cite{GEMM,ICASSP2023} do not include any constraints other than the QCE constraint, while our problem involves a large number of linear constraints. These linear constraints exclude many QCE points from the feasible set and destroy the symmetry of set $\text{conv}(\mathcal{X}_L^\RR)$, making the feasible region of problem \eqref{eqn:problem4} more complicated and the equivalence between the original and penalty problems more challenging to be established.


Up to now, we have transformed the original discrete model \eqref{eqn:problem3} into an equivalent continuous penalty model \eqref{eqn:problem4}, which is more amenable to the algorithmic design. However, problem \eqref{eqn:problem4} is still a nonconvex problem with a complicated objective function and many linear constraints.  In the following subsection, we propose an efficient inexact ALM algorithm for solving problem \eqref{eqn:problem4}.
\subsection{Inexact ALM Algorithm for Solving \eqref{eqn:problem4}}\label{sec:alm}
In this subsection, we propose an inexact ALM algorithm for solving the penalty model \eqref{eqn:problem4}.  
In Section \ref{sec:reformulation}, we first  reformulate problem \eqref{eqn:problem4} into a form that facilitates the use of the ALM framework. Then, we present the proposed inexact ALM algorithm and prove its convergence in Sections \ref{sec:inALM} and \ref{sec:converge}, respectively.
\subsubsection{Reformulation of \eqref{eqn:problem4}}\label{sec:reformulation}
For problem \eqref{eqn:problem4}, we  introduce auxiliary variables $\bw_{t,q}=\bA_q\x_t^\RR\in\mathbb{R}^2$  to simplify its objective function and introduce auxiliary variables $\bz_t\geq\mathbf{0}$ to transform its inequality linear constraints into  equality ones as follows: 
\begin{equation}\label{eqn:problem5}
\begin{aligned}
\min_{\{\x_t^\RR\},\{\bw_t\},\{\bz_t\}}~&\sum_{q=1}^Q\left(\sum_{t=1}^T\|\bw_{t,q}\|^2\right)^2\hspace{-0.22cm}\\&-\hspace{-0.1cm}\left(\sum_{q=1}^Q\sum_{t=1}^Tc_q\|\bw_{t,q}\|^2\hspace{-0.05cm}\right)^2\hspace{-0.2cm}-\hspace{-0.1cm}\lambda\hspace{-0.05cm}\sum_{t=1}^T\|\x_t^\RR\|^2\\
\text{s.t. }~~~~~~&\bC_t\x_t^\RR-\bz_t= \bb_t,~\bz_t\geq\mathbf{0},~\forall~t\in[T],\\
&\tilde{\bA}\x_t^\RR-\bw_t=\mathbf{0},~\forall~t\in[T],\\
&\x_{t,n}^\RR\in\text{conv}(\mathcal{X}_L^\RR),\,\forall\, n\in[N],\forall\, t\in[T],
\end{aligned}
\end{equation}
where $\tilde{\bA}=[\bA_1^\mathsf{T},\bA_2^\mathsf{T},\dots,\bA_Q^\mathsf{T}]^\mathsf{T}\in\mathbb{R}^{2Q\times2N}$ \vspace{0.05cm}and $\bw_t=[\bw_{t,1}^\mathsf{T},\bw_{t,2}^\mathsf{T},\dots,\bw_{t,Q}^\mathsf{T}]^\mathsf{T}\in\mathbb{R}^{2Q\times 1}$. Compared to \eqref{eqn:problem4}, the  ALM subproblem of \eqref{eqn:problem5} is more tractable, which will become more clear in Section \ref{sec:BSUM}.

For ease of presentation, we further express  problem \eqref{eqn:problem5} into the following more compact form: 
\begin{equation}\label{eqn:problem6}
\begin{aligned}
\min_{\x,\bw,\bz}~&f(\bw)+g(\bw)+h(\bx)\\
\text{s.t. }~~&\bC\x-\bz= \bb,\\
&\bA\x-\bw=\mathbf{0},\\
&\x\in\mathcal{X},~\bz\in\mathcal{Z}.
\end{aligned}
\end{equation}
where $\x=[(\x_1^\RR)^\mathsf{T},(\x_2^\RR)^\mathsf{T},\dots,(\x_T^{\RR})^\mathsf{T}]^\mathsf{T},~\bw=[\bw_1^\mathsf{T}, \bw_2^\mathsf{T},$ $\dots,\bw_T^\mathsf{T}]^\mathsf{T}, ~\bz=[\bz_1^\mathsf{T},\bz_2^\mathsf{T},\dots,\bz_T^\mathsf{T}]^\mathsf{T}$, $\bb=[\bb_1^\mathsf{T},\bb_2^\mathsf{T},\dots,\bb_T^\mathsf{T}]^\mathsf{T}$, $\bC=\text{diag}(\bC_1,\bC_2,\dots,\bC_T)$, $\bA=\mathbf{I}_T\otimes \tilde{\bA}$, $\mathcal{X}=\left\{\x\mid\x_{t,n}^\RR\in\text{conv}(\mathcal{X}_L^\RR),~\forall~ n\in[N],~\forall~t\in[T]\right\},$ $\mathcal{Z}=\{\bz\mid\bz\geq\mathbf{0}\}$, and $f(\bw),$ $g(\bw),$ and $h(\x)$ are the first, second, and third terms of the objective function in \eqref{eqn:problem5}, respectively.

 \subsubsection{Inexact ALM Algorithm}\label{sec:inALM}
 The ALM \cite{nonlinear} is a popular and powerful method for solving (equality) constrained optimization problems like \eqref{eqn:problem6}. It transforms the complicated constrained  problem into a sequence of relatively easy unconstrained subproblems or subproblems with easy constraints (e.g., nonnegative constraints) via the augmented Lagrangian function. Then, by iteratively solving the ALM subproblems and updating the Lagrange multipliers and penalty parameters, the algorithm converges to a (stationary) solution of the original constrained problem.  Next, we propose an inexact ALM algorithm for solving problem \eqref{eqn:problem6}, which \emph{inexactly} solves the ALM subproblem at each iteration in order to reduce the computational cost.
 
The augmented Lagrangian function of \eqref{eqn:problem6} is given by 
\begin{equation}\label{alf}
\begin{aligned}
\hspace{-0.2cm}\mathcal{L}_{\boldsymbol{\rho}}(\x,\bw,\bz;\bmu,\bnu)=&f(\bw)+g(\bw)+h(\bx)\\&+\bmu^\mathsf{T}(\bC\x-\bz-\bb)+\bnu^\mathsf{T}(\bA\x-\bw)\\&+\frac{\rho_\mu}{2}\|\bC\x-\bz-\bb\|^2+\frac{\rho_\nu}{2}\|\bA\x-\bw\|^2,
\end{aligned}
\end{equation}  where $(\bmu,\bnu)$ are the Lagrange multipliers corresponding to the first and second equality constraints in \eqref{eqn:problem6} and $\boldsymbol{\rho}:=(\rho_\mu,\rho_\nu)$ are the corresponding penalty parameters. The proposed inexact ALM algorithm iteratively performs the following three steps.

\textbf{Step 1: Inexactly solve the ALM subproblem over primal variables $(\x,\bw,\bz)$}. 
In the classical ALM algorithm, one needs to solve the following ALM subproblem
\begin{equation}\label{ALMsubproblem}
\hspace{-0.2cm}(\x^{m+1}\hspace{-0.05cm},\bw^{m+1}\hspace{-0.05cm},\bz^{m+1})\hspace{-0.05cm}\in\arg\hspace{-0.35cm}\min_{\x\in\mathcal{X},\bw,\bz\in\mathcal{Z}}\hspace{-0.1cm}\mathcal{L}_{\boldsymbol{\rho}^m}(\x,\bw,\bz;\bmu^{m}\hspace{-0.05cm},\bnu^{m})
\end{equation}  at each iteration \cite{nonlinear}. To avoid the high computational cost of exactly solving \eqref{ALMsubproblem}, the  proposed inexact ALM algorithm only   requires \eqref{ALMsubproblem} to be solved to an $\epsilon_m$-stationary point. Here $(\x^{m+1},\bw^{m+1},\bz^{m+1})$ is called an $\epsilon_m$-stationary point of the ALM subproblem if
 \begin{equation}\label{inexactcondition}
\begin{aligned}
&\text{dist}\,(\nabla\mathcal{L}_{\boldsymbol{\rho}^m}(\x^{m+1},\bw^{m+1},\bz^{m+1};\bmu^{m},\bnu^{m})\\
&\qquad\qquad\quad+\partial \mathbb{I}_{\mathcal{X}}(\x^{m+1})+\partial\mathbb{I}_{\mathcal{Z}}(\bz^{m+1}),\mathbf{0})\leq \epsilon_{m},
\end{aligned}
\end{equation}
where $\mathbb{I}_{\mathcal{X}}$ and $\mathbb{I}_{\mathcal{Z}}$ denote the indicator functions of $\mathcal{X}$ and $\mathcal{Z}$, respectively, and $\partial \mathbb{I}_{\mathcal{X}}$ and $\partial \mathbb{I}_{\mathcal{Z}}$ denote their subdifferentials, respectively. When $\epsilon_m=0$ in \eqref{inexactcondition}, then $(\x^{m+1},\bw^{m+1}, \bz^{m+1})$ is a stationary point of problem \eqref{ALMsubproblem}. To guarantee the convergence of the inexact ALM algorithm, the accuracy of solving the ALM subproblem should be gradually increased.

In our proposed algorithm, the ALM subproblem \eqref{ALMsubproblem} is (approximately) solved by a custom-built BSUM algorithm, which admits  closed-form updates and is guaranteed to find an $\epsilon_m$-stationary point within a finite number of iterations. To maintain coherence in  presenting the  proposed inexact ALM algorithm, we defer detailed discussions of the BSUM algorithm for solving \eqref{ALMsubproblem} to Section \ref{sec:BSUM}.

\textbf{Step 2: Update the Lagrange multipliers $(\bmu,\bnu)$.} After obtaining $(\x^{m+1}, \bw^{m+1}, \bz^{m+1})$, the Lagrange multipliers $(\bmu,\bnu)$ are updated as 

{\small\begin{equation}\label{multiplier}
\begin{aligned}
\bmu^{m+1}&\hspace{-0.05cm}=\min\hspace{-0.05cm}\left\{\max\{\bmu^{m}\hspace{-0.05cm}+\rho_\mu^m(\bC\x^{m+1}\hspace{-0.05cm}-\bz^{m+1}\hspace{-0.05cm}-\bb),\mu_{\min}\},\mu_{\max}\right\}\hspace{-0.05cm},\\
\bnu^{m+1}&\hspace{-0.05cm}=\min\hspace{-0.05cm}\left\{\max\{\bnu^{m}\hspace{-0.05cm}+\rho_\nu^m(\bA\x^{m+1}-\bw^{m+1}),\nu_{\min}\},\nu_{\max}\right\},
\end{aligned}
\end{equation}}

\hspace{-0.35cm}where $(\mu_{\min},\nu_{\min})$ and $(\mu_{\max},\nu_{\max})$ are the imposed lower and upper bounds on the Lagrange multipliers $(\bmu,\bnu)$, respectively, and the operators $\min$ and $\max$ both operate component-wise on their entries. Different from the classical ALM, we apply safeguards in our method to ensure that the Lagrange multipliers are bounded, which is crucial for the convergence of the algorithm \cite{boundALM}.

\textbf{Step 3: Update the penalty parameters $\boldsymbol{\rho}=(\rho_\mu,\rho_\nu)$.} Finally, we update the penalty parameters $\boldsymbol{\rho}=(\rho_\mu,\rho_\nu)$ as in \eqref{rho}, which is given on the top of the next page, where $\tau>1$ and $\delta>0$. \begin{figure*}\begin{equation}\label{rho}
 \boldsymbol{\rho}^{m+1}=\left\{
\begin{aligned}
\tau\boldsymbol{\rho}^m,~&\text{if } \left\|\begin{smallmatrix}\sqrt{\rho^m_\mu}(\bC\x^{m+1}-\bz^{m+1}-\bb)\\ \sqrt{\rho^m_\nu}(\bA\x^{m+1}-\bw^{m+1})\end{smallmatrix}\right\| \geq \delta\left\|\begin{smallmatrix}\sqrt{\rho^m_\mu}(\bC\x^{m}-\bz^{m}-\bb)\\ \sqrt{\rho^m_\nu}(\bA\x^{m}-\bw^{m})\end{smallmatrix}\right\|;\\
\boldsymbol{\rho}^m,\hspace{0.2cm}~&\text{otherwise}.
 \end{aligned}\right.
 \end{equation}
 \vspace{-0.1cm}
 \hrule
 \vspace{-0.3cm}
\end{figure*}
The criterion in \eqref{rho} increases the penalty parameters if the violation of the constraints is not sufficiently reduced and keep them fixed otherwise. We remark that the penalty parameters $\rho_\mu$ and $\rho_\nu$ should be increased at the same rate as in \eqref{rho} to ensure the convergence of the algorithm.

We summarize the proposed inexact ALM algorithm for solving problem \eqref{eqn:problem6} in Algorithm \ref{alg:alm}.

\begin{algorithm}[t]
\caption{Inexact ALM algorithm for solving \eqref{eqn:problem6}}\label{alg:alm}
	\small
	\begin{algorithmic}[1]
		\STATE \textbf{Input:} $(\bx^0,\bw^0,\bz^0), (\bmu^0, \bnu^0)$, $\boldsymbol{\rho}^0=(\rho_\mu^0,\rho_\nu^0)$, $-\infty<\mu_{\min}<\mu_{\max}<\infty$, $-\infty<\nu_{\min}<\nu_{\max}<\infty$,  $\tau>1, \delta>0$, a positive sequence $\epsilon_m\to 0$.
		
\STATE \textbf{Initialize}: $m=0$.
\REPEAT
  \STATE\emph{(Inexactly solve the ALM subproblem)}: Obtain an $\epsilon_m$-stationary point $(\x^{m+1},\bw^{m+1},\bz^{m+1})$ of \eqref{ALMsubproblem} by Algorithm \ref{alg:BSUM} (see further ahead).
		\STATE  \emph{(Update the Lagrange multipliers)}: Obtain $(\bmu^{m+1},\bnu^{m+1})$ by \eqref{multiplier}. 
\STATE  \emph{(Update the penalty parameter)}: Obtain $\boldsymbol{\rho}^{m+1}$  by \eqref{rho}.\\
				\STATE Set  $m = m+1.$
		\UNTIL some stopping criterion is satisfied.
		\STATE \textbf{Output:} $(\x^m,\bw^m,\bz^m)$.
	\end{algorithmic} 
	\end{algorithm}

\subsubsection{Convergence Analysis} \label{sec:converge}
Before showing the convergence guarantee of the proposed inexact ALM  algorithm, we first give the formal definition of the stationary point  of problem \eqref{eqn:problem6} \cite[Proposition 2.1.2]{nonlinear}.
\begin{definition}[Stationary point of \eqref{eqn:problem6}]
A point $(\hat{\bx},\hat{\bw},\hat{\bz})$ is called a stationary point of \eqref{eqn:problem6} if it satisfies  $(\hat{\x},\hat{\bw},\hat{\bz})\in\mathcal{F}$ and 
\begin{equation}\label{optcondition}
\begin{aligned}
(\nabla f(\bw)+&\nabla g (\bw))^\mathsf{T}(\bw-\hat{\bw})+\nabla h(\x)^\mathsf{T} (\x-\hat{\x})\geq~0,\\&\forall ~(\x,\bw)\in\{(\x,\bw)\mid\exists~ \bz, (\x,\bw,\bz)\in\mathcal{F}\},
\end{aligned}
\end{equation}
where $$\mathcal{F}:=\left\{(\x,\bw,\bz)\mid \bC\x-\bz=\bb, \bA\x-\bw=\mathbf{0}, \x\in\mathcal{X}, \bz\in\mathcal{Z}\right\}$$ is the feasible region of \eqref{eqn:problem6}.
\end{definition}
 Now we present the convergence result of Algorithm \ref{alg:alm} in the following Theorem  \ref{theorem1}. .
\begin{theorem}\label{theorem1}
Any limit point $(\bar{\x},\bar{\bw},\bar{\bz})$ of $\{(\x^m,\bw^m,\bz^m)\}$ generated by Algorithm \ref{alg:alm} is a stationary point of \eqref{eqn:problem6}. 
 Moreover, $$\lim_{m\to\infty}\text{\normalfont{dist}}((\bx^m,\bw^m,\bz^m),\mathcal{S}^*)=0,$$ where $\mathcal{S}^*$ is the set of stationary points of \eqref{eqn:problem6}.
\end{theorem}
\begin{proof}
See Appendix \ref{appendixC}.
\end{proof}
We remark that the above  result highly depends on the boundedness of the sequences $\{(\x^m,\bw^m,\bz^m)\}$ and $\{(\bmu^m,\bnu^m)\}$ generated by Algorithm \ref{alg:alm}. In sharp contrast to classical convergence results for the ALM where the boundedness is typically imposed as an uncheckable assumption \cite{nonlinear}, in the proposed inexact ALM algorithm, the boundedness of the Lagrange multipliers is naturally ensured by the update rule in \eqref{multiplier}, and the boundedness of the primal variables is guaranteed by the specialized BSUM algorithm  for solving the ALM subproblem; see Section \ref{sec:BSUM} and Appendix \ref{appendixB} for detailed discussions. 

It is  worth mentioning that a similar inexact ALM framework has been proposed in \cite{boundALM}.  However, the problem considered in \cite{boundALM} is much simpler than \eqref{eqn:problem6}. In addition, the condition that guarantees the convergence in \cite{boundALM} is more restrictive than that in the proposed inexact ALM algorithm. In particular, in \cite{boundALM}, $\epsilon_m$ is  required to satisfy $\epsilon_m\rho^m\to0$    (where $\rho^m$ is the corresponding penalty parameter), while we only require $\epsilon_m\to 0$.

\subsection{BSUM Algorithm for Solving Subproblem \eqref{ALMsubproblem}}\label{sec:BSUM}
The efficiency of the proposed inexact ALM algorithm is mainly determined by the efficiency of solving the ALM subproblem \eqref{ALMsubproblem} to obtain an $\epsilon_m$-stationary point. In this subsection, we propose an efficient BSUM algorithm for solving \eqref{ALMsubproblem} by judiciously exploiting the problem's special structure.
Since $(\bmu^{m},\bnu^{m})$ are treated as fixed parameters when solving \eqref{ALMsubproblem},  we denote 
\begin{equation}\label{def:lm}
\mathcal{L}_m(\x,\bw,\bz):=\mathcal{L}_{\boldsymbol{\rho}^m}(\x,\bw,\bz;\bmu^{m},\bnu^{m})
\end{equation} in this subsection for simplicity of notations.

Noting that the constraints in  \eqref{ALMsubproblem} are separable within each block of variables $\{\x,\bw,\bz\}$, it is convenient to update the variables in an alternating fashion to solve the problem.  However, the objective function is complicated (even with respect to each block of variables) and cannot be efficiently optimized within each block. This motivates us to apply a BSUM framework, which minimizes an  upper bound of the objective function for each block at each iteration.  Specifically, the BSUM framework is given by 
 \begin{subequations}
 \begin{align}
 \x^{(r+1)}&\in\arg\min_{\x\in\mathcal{X}}~u_x(\x;\x^{(r)},\bw^{(r)},\bz^{(r)});\label{bsum1}\\
  \bw^{(r+1)}&\in\arg\min~u_w(\bw;\x^{(r+1)},\bw^{(r)},\bz^{(r)});\\
    \bz^{(r+1)}&\in\arg\min_{\bz\in\mathcal{Z}}~u_z(\bz;\x^{(r+1)},\bw^{(r+1)},\bz^{(r)})\label{bsum3},
 \end{align}
 \end{subequations}
 where $u_x, u_w,$ and $u_z$ are locally tight upper bounds of $\mathcal{L}_m$ in \eqref{def:lm}, i.e., $u_x(\x;\x^{(r)},\bw^{(r)},\bz^{(r)})\geq \mathcal{L}_m(\x,\bw^{(r)},\bz^{(r)})$ for all $\x\in\mathcal{X}$ with equality  holding at $\x=\x^{(r)}$ (similar properties hold for $u_w$ and $u_z$). Please see \cite{BSUM} for more detailed discussions on BSUM. Here we use $\{(\x^{(r)},\bw^{(r)},\bz^{(r)})\}$ to denote the inner sequence generated by the BSUM framework for solving the ALM subproblem, which is different from the outer sequence $\{(\bx^m,\bw^m,\bz^m)\}$ generated by the inexact ALM.   
 
In the following, we tailor a BSUM algorithm for solving problem \eqref{ALMsubproblem}. The key is to design the upper bounds $u_x,u_w,u_z$ appropriately such that they are good approximations of the objective function $\mathcal{L}_m$ and at the same time the subproblems \eqref{bsum1}--\eqref{bsum3} admit closed-form solutions or can be efficiently solved. Next we discuss the updates of the $\x$-, $\bw$-, and $\bz$-blocks separately. 

 \textbf{Update of variable $\x$}. The difficulty in updating variable $\x$ mainly arises from the complicated constraint $\x\in\mathcal{X}$ defined below \eqref{eqn:problem6}. Notably,  the set $\mathcal{X}$ is decoupled over $\{\x_{t,n}^\RR\}$. To leverage this property, we construct a separable upper bound  function of $\mathcal{L}_m(\x,\bw^{(r)},\bz^{(r)})$ as follows: 
\begin{equation}\label{ux}
\begin{aligned}&u_x(\x; \x^{(r)},\bw^{(r)},\bz^{(r)})= \mathcal{L}_m(\x^{(r)},\bw^{(r)},\bz^{(r)})\\&\quad~+\nabla_\x \mathcal{L}_m(\x^{(r)},\bw^{(r)},\bz^{(r)})^\mathsf{T}(\x-\x^{(r)})+\frac{\gamma^m}{2}\|\x-\x^{(r)}\|^2.\\
\end{aligned}
\end{equation}
The above $u_x(\x; \x^{(r)},\bw^{(r)},\bz^{(r)})$  is exactly the first-order Taylor expansion of $\mathcal{L}_m(\x, \bw^{(r)}, \bz^{(r)})$ at $\x=\x^{(r)}$ plus a quadratic regularization term, which equals $\mathcal{L}_m(\x,\bw^{(r)},\bz^{(r)})$ at $\x=\x^{(r)}$. To ensure that it serves as an upper bound of $\mathcal{L}_m(\x, \bw^{(r)}, \bz^{(r)})$,  we set $\gamma^m= \|\rho^m_\mu\bC^\mathsf{T}\bC+\rho_\nu^m{\bA}^\mathsf{T}{\bA}\|$. 

With $u_x(\x;\x^{(r)},\bw^{(r)},\bz^{(r)})$ defined in \eqref{ux} and by ignoring the constant terms, the $\x$-subproblem can be expressed as 
\begin{equation}\label{updateofx}
\begin{aligned}
\min_{\x\in\mathcal{X}}\nabla_{\x}\mathcal{L}_{m}(\x^{(r)}\hspace{-0.05cm},\bw^{(r)}\hspace{-0.05cm},\bz^{(r)})^{\mathsf{T}}\x\hspace{-0.05cm}+\hspace{-0.07cm}\frac{\gamma^m}{2}\|\x-\x^{(r)}\|^2\hspace{-0.05cm},
\end{aligned}
\end{equation}
whose solution is given by
\begin{equation}\label{xsolution}
\begin{aligned}
\x^{(r+1)}=&\mathcal{P}_{\mathcal{X}}\left(\x^{(r)}-\frac{\nabla\,\mathcal{L}_{m}(\x^{(r)},\bw^{(r)},\bz^{(r)})}{\gamma^m}\right).
\end{aligned}
\end{equation}
 According to the definition of  $\mathcal{X}$, the projection onto $\mathcal{X}$ can be decomposed into $NT$ of projections onto $\text{conv}(\mathcal{X}_L^\RR)$, which admit closed-form expressions \cite{GEMM}.
   
  \textbf{Update of variable $\bw$.} The functions $f(\bw)$ and $g(\bw)$ involved in the ALM subproblem are both quartic in $\bw$, making the $\bw$-update challenging.  In particular, $g(\bw)$ is concave and couples all the components in $\bw$. To simplify the objective function while best preserving its structure, we linearize the concave quartic term $g(\bw)$ and keep all the convex terms unchanged, which gives the following locally tight upper bound of $\mathcal{L}_m(\x^{(r+1)},\bw,\bz^{(r)})$:
  \begin{equation*}
\begin{aligned}
&u_w(\bw;\bx^{(r+1)},\bw^{(r)},\bz^{(r)})\\
=\,&f(\bw)+g(\bw^{(r)})+\nabla g(\bw^{(r)})^\mathsf{T}(\bw-\bw^{(r)})+h(\bx^{(r+1)})\\&+\hspace{-0.05cm}(\bmu^{m})^\mathsf{T}(\bC\x^{(r+1)}\hspace{-0.06cm}-\bz^{(r)}-\bb)+(\bnu^{m})^\mathsf{T}(\bA\x^{(r+1)}\hspace{-0.05cm}-\bw)\\&+\frac{\rho_\mu^{m}}{2}\|\bC\x^{(r+1)}-\bz^{(r)}-\bb\|^2+\frac{\rho_\nu^m}{2}\|\bA\x^{(r+1)}-\bw\|^2.
\end{aligned}
\end{equation*} 
Then, by omitting the constant terms, the $\bw$-subproblem reads
 \begin{equation}\label{wupdate}
 \begin{aligned}
\min_\bw \,f(\bw)+(\nabla g(\bw^{(r)})-\bnu^{m})^\mathsf{T}\bw+\frac{\rho_v^{m}}{2}\|\bA\x^{(r+1)}-\bw\|^2.
 \end{aligned}
 \end{equation}
Note that in the objective function of \eqref{wupdate}, $f(\bw)$, i.e., the first term of the objective function of \eqref{eqn:problem5},  is separable over the index $q\in[Q]$, and both of the last two terms  are fully separable. Hence, the problem can be decomposed into $Q$ individual subproblems. 

  Let $\bw_{(q)}:=[\bw_{1,q}^\mathsf{T},\bw_{2,q}^\mathsf{T},\dots,\bw_{T,q}^\mathsf{T}]^\mathsf{T}$ denote the vector collecting all the components in $\bw$ related to $q$ and let $\boldsymbol{\xi}_q^{(r)}$ be the coefficient of the linear term in \eqref{wupdate} related to $\bw_{(q)}$.
   Then for each  $q\in[Q]$, the $\bw_{(q)}$-subproblem has the following form:
\begin{equation}\label{wq2}\arg\min_{\bw_{(q)}}\|\bw_{(q)}\|^4+\frac{\rho_v^m}{2}\|\bw_{(q)}\|^2+(\boldsymbol{\xi}_{{q}}^{(r)})^\mathsf{T}\bw_{(q)}.\end{equation}
Since the first two terms in the objective function of \eqref{wq2} are irrelevant to the angle of $\bw_{(q)}$, it follows that the optimal angle of $\bw_{(q)}$ is $-\boldsymbol{\xi}_{{q}}^{(r)}/\|\boldsymbol{\xi}_{{q}}^{(r)}\|$.  Therefore, we only need to optimize the norm of $\bw_{(q)}$, i.e., solving a univariate quartic problem as follows: 
 \begin{equation}\label{normw}
 \beta_q^{(r+1)}\in\arg\min_{\beta\geq 0}~\left\{\beta^4+\frac{\rho_v^m}{2}\beta^2-\left\|\boldsymbol{\xi}_{{q}}^{(r)}\right\|\beta\right\}.
 \end{equation} 
By setting the derivative of the quartic objective function in \eqref{normw} to zero, we obtain the following cubic equation:
   $$ 4\beta^3+\rho_v^m\beta-\left\|\boldsymbol{\xi}_{{q}}^{(r)}\right\|=0.$$
The solution of \eqref{normw} is the unique positive root of the above equation, whose explicit expression is given by \vspace{-0.2cm} 
$$\begin{aligned}\beta_q^{(r+1)}=&\left(\frac{\|\boldsymbol{\xi}_{{q}}^{(r)}\|}{8}+\left(\left(\frac{\|\boldsymbol{\xi}_{{q}}^{(r)}\|}{8}\right)^2+\left(\frac{\rho_\nu^m}{12}\right)^3\right)^{\frac{1}{2}}\right)^{\frac{1}{3}}\\&+\left(\frac{\|\boldsymbol{\xi}_{{q}}^{(r)}\|}{8}-\left(\left(\frac{\|\boldsymbol{\xi}_{{q}}^{(r)}\|}{8}\right)^2+\left(\frac{\rho_\nu^m}{12}\right)^3\right)^{\frac{1}{2}}\right)^{\frac{1}{3}}.\end{aligned}$$ 
Combining the above discussions, we obtain the optimal solution of \eqref{wq2}: \vspace{-0.2cm}
\begin{equation}\label{wsolution}
\bw_{(q)}^{(r+1)}=-\frac{\beta_q^{(r+1)}}{\|\boldsymbol{\xi}_{{q}}^{(r)}\|} \boldsymbol{\xi}_{{q}}^{(r)},~\forall~q\in[Q].
\end{equation}

 \textbf{Update of variable $\bz$}. The ALM subproblem with respect to  $\bz$ is  straightforward to solve. In particular,  we simply set 
 $$u_z(\bz; \bx^{(r+1)}, \bw^{(r+1)}, \bz^{(r)})=\mathcal{L}_m(\bx^{(r+1)},\bw^{(r+1)},\bz).$$
The  $\bz$-subproblem\vspace{-0.1cm}
 \begin{equation}\label{updatez}
\min_{\bz\in\mathcal{Z}}\,-(\bmu^{m})^\mathsf{T}\bz+\frac{\rho_\mu^m}{2}\|\bz-\bC\x^{(r+1)}+\bb\|^2
\end{equation}
 admits a closed-form solution as 
\begin{equation}\label{zsolution}
\bz^{(r+1)}=\mathcal{P}_{\{\bz\geq\mathbf{0}\}}\left(\bC\x^{(r+1)}-\bb+\frac{\bmu^{m}}{\rho_\mu^m}\right).
\end{equation}

The proposed BSUM algorithm for solving the ALM subproblem \eqref{ALMsubproblem} is summarized in Algorithm \ref{alg:BSUM}. Since all variables admit closed-form updates, the BSUM algorithm can be performed  efficiently. Note that in Algorithm \ref{alg:BSUM}, we set the initial point $(\x^{(0)}, \bw^{(0)}, \bz^{(0)})$ as the approximate stationary point of the previous ALM subproblem. This warm-start strategy is able to accelerate the convergence of BSUM, thereby enhancing the numerical  efficiency of the inexact ALM algorithm.  As will be demonstrated in the simulation, Algorithm \ref{alg:BSUM} equipped with the warm-start strategy can generally terminate within only a few iterations. 
\begin{algorithm}[t]
\caption{BSUM algorithm for solving \eqref{ALMsubproblem}}\label{alg:BSUM}
	\small
	\begin{algorithmic}[1]
		\STATE \textbf{Input:} $(\x^{m}, \bw^{m}, \bz^{m})$, $(\bmu^{m}, \bnu^{m}), \boldsymbol{\rho}^m=(\rho_\mu^m,\rho_\nu^m), $ $\gamma^m= \|\rho^m_\mu\bC^\mathsf{T}\bC+\rho_\nu^m{\bA}^\mathsf{T}{\bA}\|$, $\epsilon_m$.
		\STATE \textbf{Initialize}: $(\x^{(0)}, \bw^{(0)}, \bz^{(0)})=(\x^{m}, \bw^{m}, \bz^{m})$, $r=0$.
		                 \REPEAT
		\STATE Update $\x^{(r+1)}$, $\bw^{(r+1)}$, and $\bz^{(r+1)}$ by \eqref{xsolution}, \eqref{wsolution}, and \eqref{zsolution}, respectively.
		\STATE Set $r=r+1$.
		\UNTIL $(\x^{(r)},\bw^{(r)},\bz^{(r)})$ is an $\epsilon_m$-stationary point of \eqref{ALMsubproblem}.
		\STATE \textbf{Output:} $(\x^{m+1},\bw^{m+1},\bz^{m+1})=(\x^{(r)},\bw^{(r)},\bz^{(r)})$.
	\end{algorithmic} \label{nl1p}
	
	\vspace{-0.05cm}
\end{algorithm}

 In the following proposition, we theoretically show  that Algorithm \ref{alg:BSUM} can terminate within a finite number of iterations, i.e.,  an $\epsilon_m$-stationary point of \eqref{ALMsubproblem} can be found within a finite number of iterations.   Note that according to the update rule of $\boldsymbol{\rho}$ in \eqref{rho}, $\rho_\mu$ and $\rho_\nu$  are increased consistently, and thus we can express $\boldsymbol{\rho}^m$ as $\boldsymbol{\rho}^m=(a_\mu,a_\nu){\rho}^m$ with $a_\mu>0$ and $a_\nu>0$ being two constants independent of $m$. With this notation, the  iteration complexity of Algorithm \ref{alg:BSUM} is given as follows.

\begin{proposition}\label{proposition2}
Algorithm \ref{alg:BSUM} can return an $\epsilon_m$-stationary point of \eqref{ALMsubproblem}  within $\mathcal{O}\left(\frac{(\rho^m)^2}{\epsilon_m^2}\right)$ iterations.
\end{proposition}
\begin{proof}
See Appendix \ref{appendixD}.
\end{proof}
\section{Numerical Results}\label{sec:4}
In this section, we present simulation results to evaluate the performance of the proposed approach under various system configurations.  Each system configuration is represented by a tuple $(N, K, M, L)$, where $N$ and $K$ are the numbers of transmit antennas and  users in the system, respectively, $M$ is the constellation order, and $L$ is the quantization level.
 We generate the simulation parameters in the same way as in \cite{DFRC2, DFRC_CI1}. Specifically, the transmit power is fixed as $P=1$ and the SNR is defined as $1/\sigma^2$. We assume a Rayleigh fading communication channel $\bH$, whose elements are independently  drawn from $\mathcal{CN}(0, 1)$. 
 We assume that there are three interested targets located at  $\bar{\theta}_1=-40^\circ,~\bar{\theta}_2=0^\circ,$ and $\bar{\theta}_3=40^\circ$, respectively, and the desired beampattern is given by  
\begin{equation}\label{desire}d(\theta)=\left\{
\begin{aligned}
&1,~~\text{if }\theta\in\left[\bar{\theta}_i-\frac{\Delta_\theta}{2},\bar{\theta}_i+\frac{\Delta_\theta}{2}\right],~i=1,2,3;\\
&0,\hspace{0.24cm}\text{otherwise},
\end{aligned}\right.
\end{equation}
where $\Delta_\theta=10^\circ$ is the beam width. The direction grids $\{\theta_q\}_{q\in[Q]}$ are uniformly sampled from $-90^\circ$ to $90^\circ$ with a resolution of $1^\circ$. The safety margin thresholds for different users at different time slots are set as the same, i.e., $b_{t,k}=b$ for all $t\in[T]$ and $k\in[K].$ The block length is set as $T=50$.

The implementation details of the proposed Algorithm \ref{alg:alm} are as follows. 
We set $\rho^0_\mu=0.01\sqrt{\lambda}$ and $\rho_\nu^0=\frac{1}{3}{\rho_\mu^0}$. Both $(\bx^0,\bw^0,\bz^0)$ and  $(\bmu^0, \bnu^0)$ are set as all-zero vectors. The bounds on the Lagrange multipliers are   given by $\mu_{\min}=\nu_{\min}=-10^3,~\mu_{\max}=\nu_{\max}=10^3$. Finally, $\tau=1.01, \delta=0.95,$ and  $\epsilon_m=\frac{1}{m}$. Algorithm \ref{alg:alm} is terminated when the iteration number is larger than $500$ or when 
\begin{equation}\label{stop}\max\{e^{m-1},\|\bC\x^m-\bz^m-\bb\|,\|\bA\x^m-\bw^m\|\}\leq \sqrt{T}\times 10^{-3}, \end{equation}
where 
$$
\begin{aligned}
e^{m-1}:=&\|\nabla_\x \mathcal{L}(\x^{m},\bw^m,\bz^m)-\nabla_\x \mathcal{L}_m(\x^{m-1},\bw^{m-1},\bz^{m-1})\\&~-\gamma^m(\x^{m}-\x^{m-1})\|+\|\nabla g(\bw^{m})-\nabla g(\bw^{m-1})\|\end{aligned}$$  is defined  to characterize the accuracy of solving the ALM subproblem (see \eqref{def:e2} and \eqref{def:e}). 
The above stopping criterion is employed based on the fact that when $\max\{{e}^{m-1},\|\bC\x^m-\bz^m-\bb\|,|\bA\x^m-\bw^m\|\}$ is small, the output of Algorithm \ref{alg:alm} approaches a stationary point of the penalty model \eqref{eqn:problem6}. 

In the following, we first show the convergence behavior of the proposed inexact ALM and BSUM algorithms  in Section \ref{sec:simulation1}. Then, we examine the impact of various problem parameters on the communication and radar performance in Section \ref{sec:simulation2}. Finally, we compare the performance of the proposed approach with existing SOTA ones in Section \ref{sec:simulation3}.
\begin{figure}
\subfigure[Inner loop.]{\includegraphics[width=0.48\columnwidth]{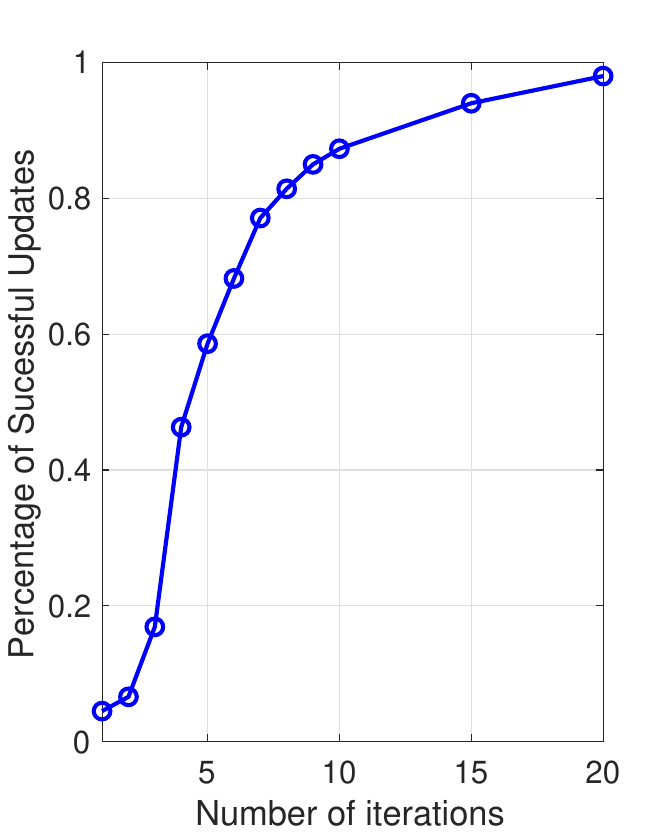}}
\subfigure[Outer loop.]{\includegraphics[width=0.48\columnwidth]{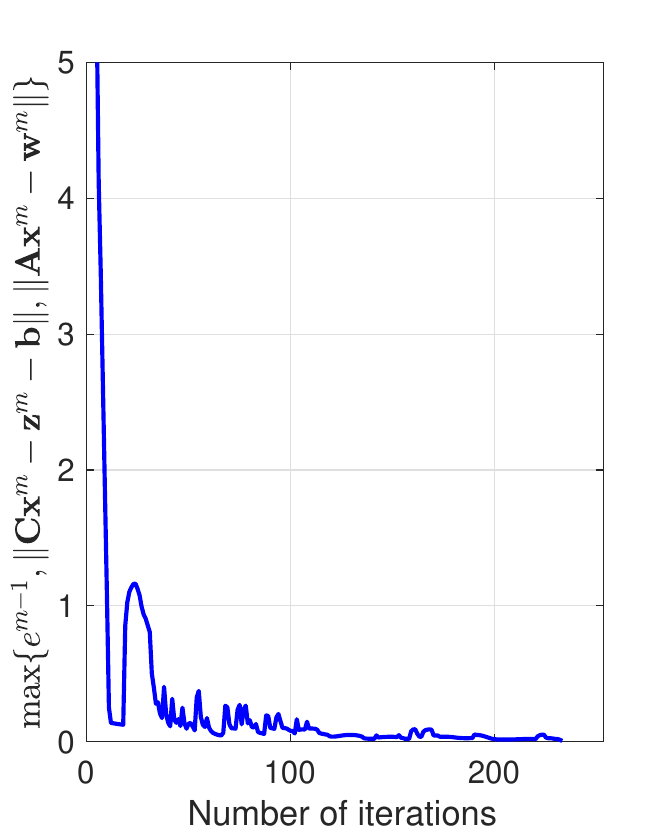}}
\centering
\vspace{-0.1cm}
\caption{Convergence behaviors of proposed Algorithms \ref{alg:alm} and \ref{alg:BSUM}.}
\label{converge}
\vspace{-0.3cm}
\end{figure}
\begin{figure}
\includegraphics[scale=0.38]{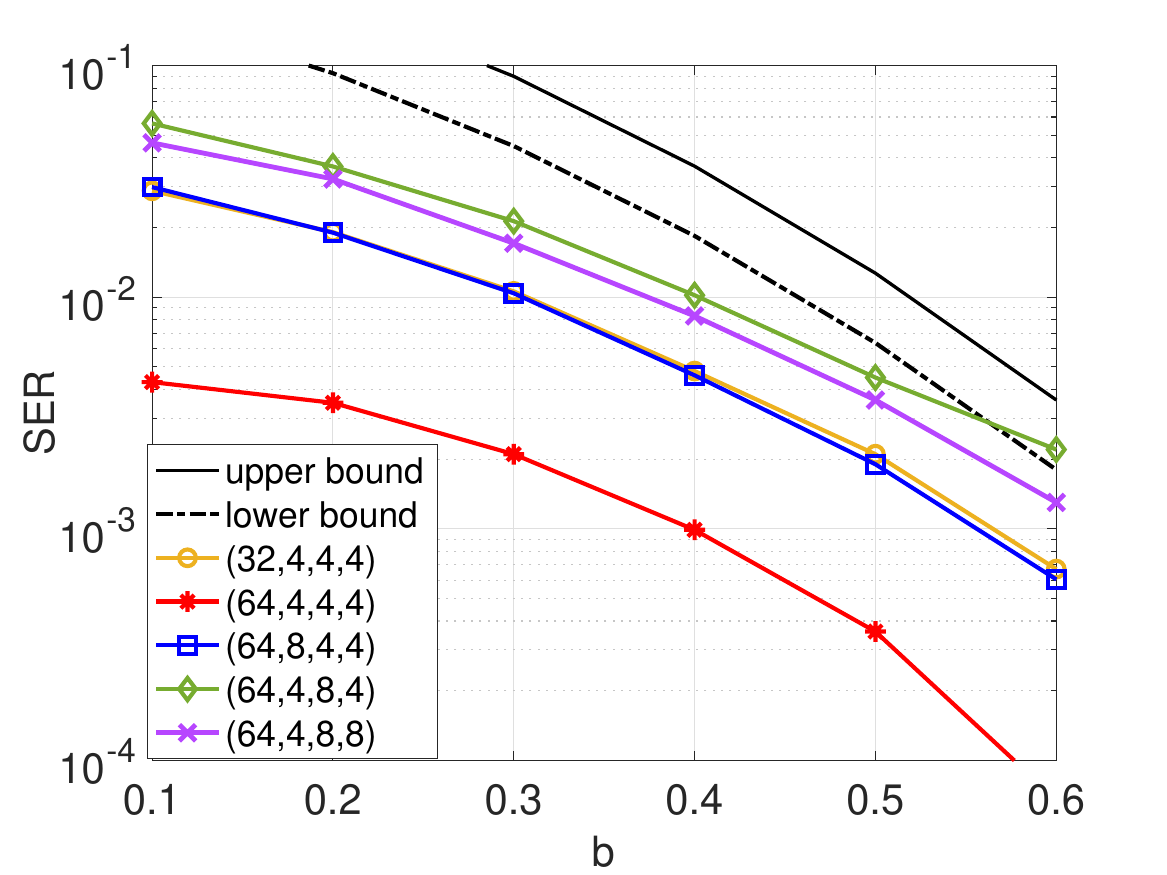}
\centering
\caption{Achieved SERs by the proposed approach versus the safety margin threshold $b$ under different system configurations, where the tuple in the legend represents the system parameter $(N,K,M,L)$ and the SNR is set as $10$ dB. The curves labeled ``upper bound'' and ``lower bound'' correspond to the theoretical upper and lower bounds on the SEP in \eqref{upperbound}, respectively.}
\label{SER_system}
\vspace{-0.4cm}
\end{figure}
\subsection{Convergence Behavior}\label{sec:simulation1}
\begin{figure*}
\subfigure[Different $N$.]{\includegraphics[width=0.5\columnwidth]{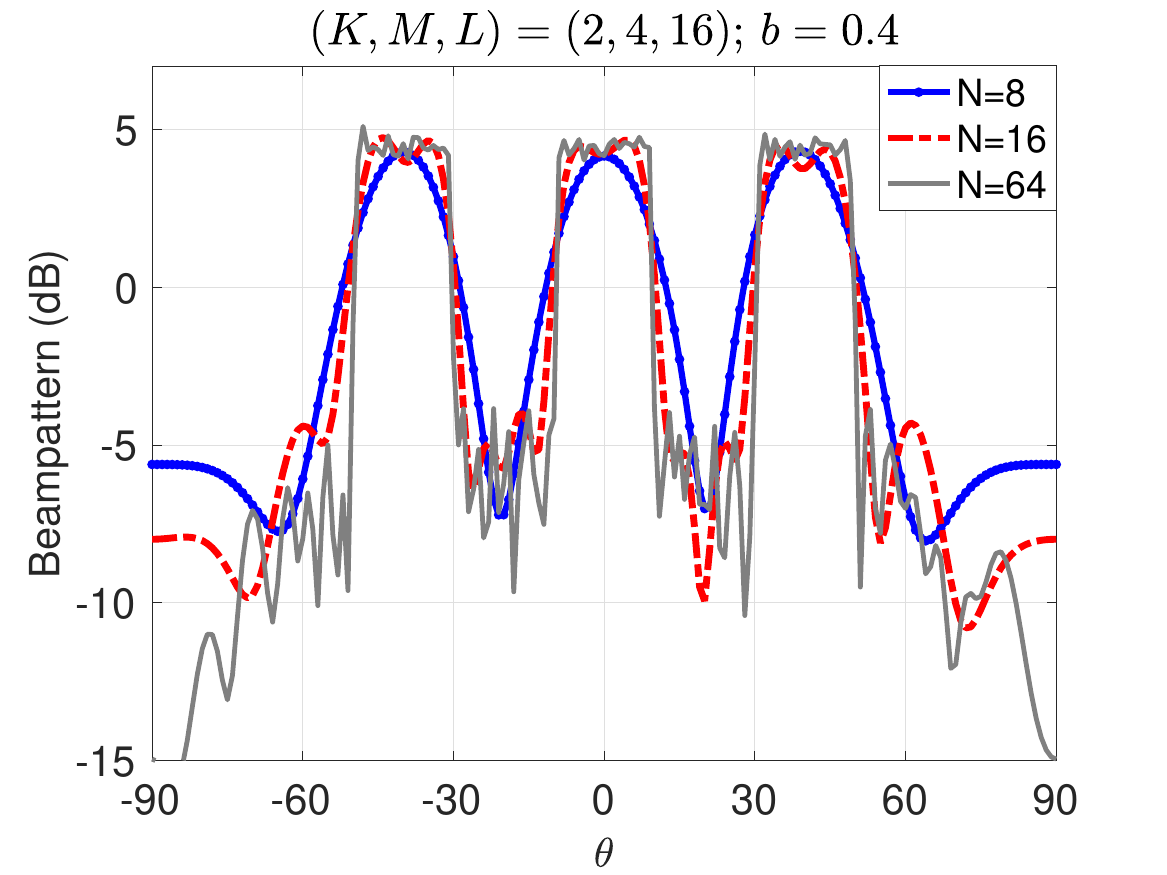}}
\subfigure[Different $K$.]{\includegraphics[width=0.5\columnwidth]{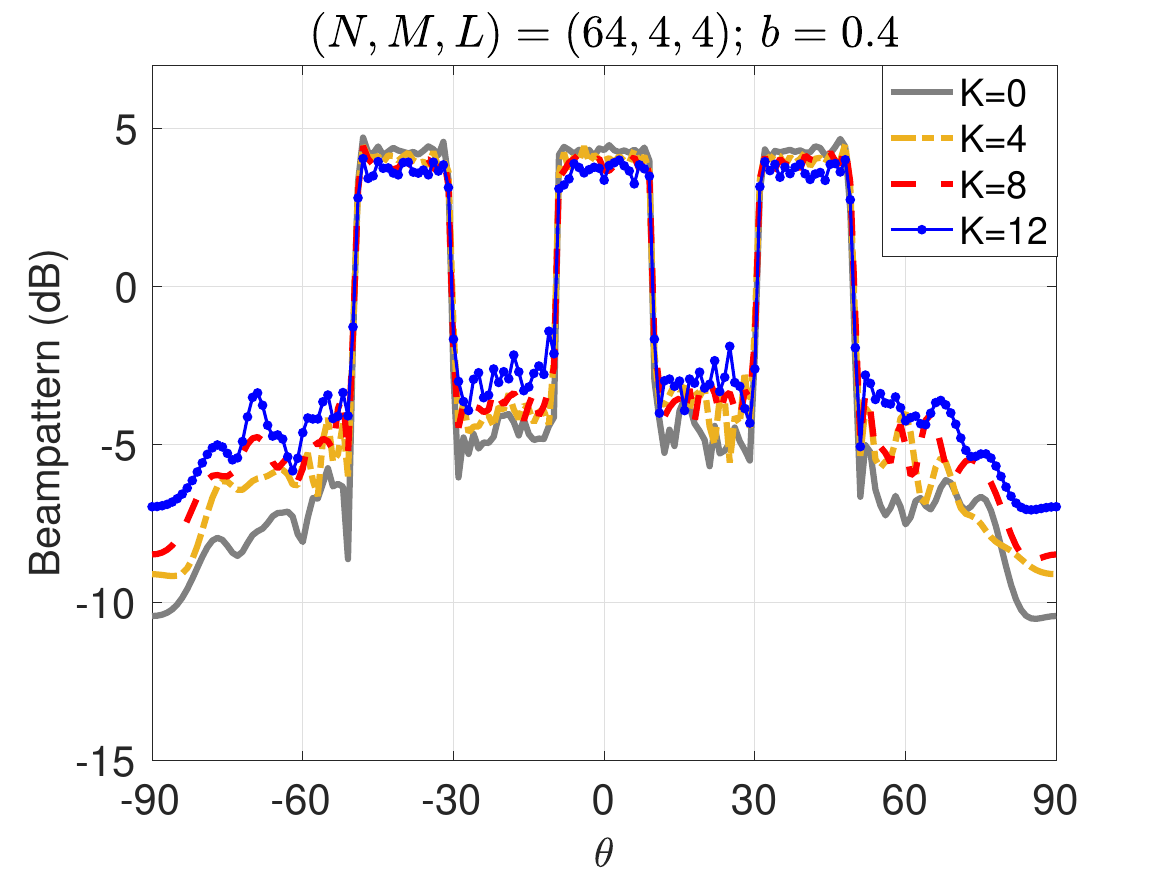}}
\subfigure[Different $L$.]{\includegraphics[width=0.5\columnwidth]{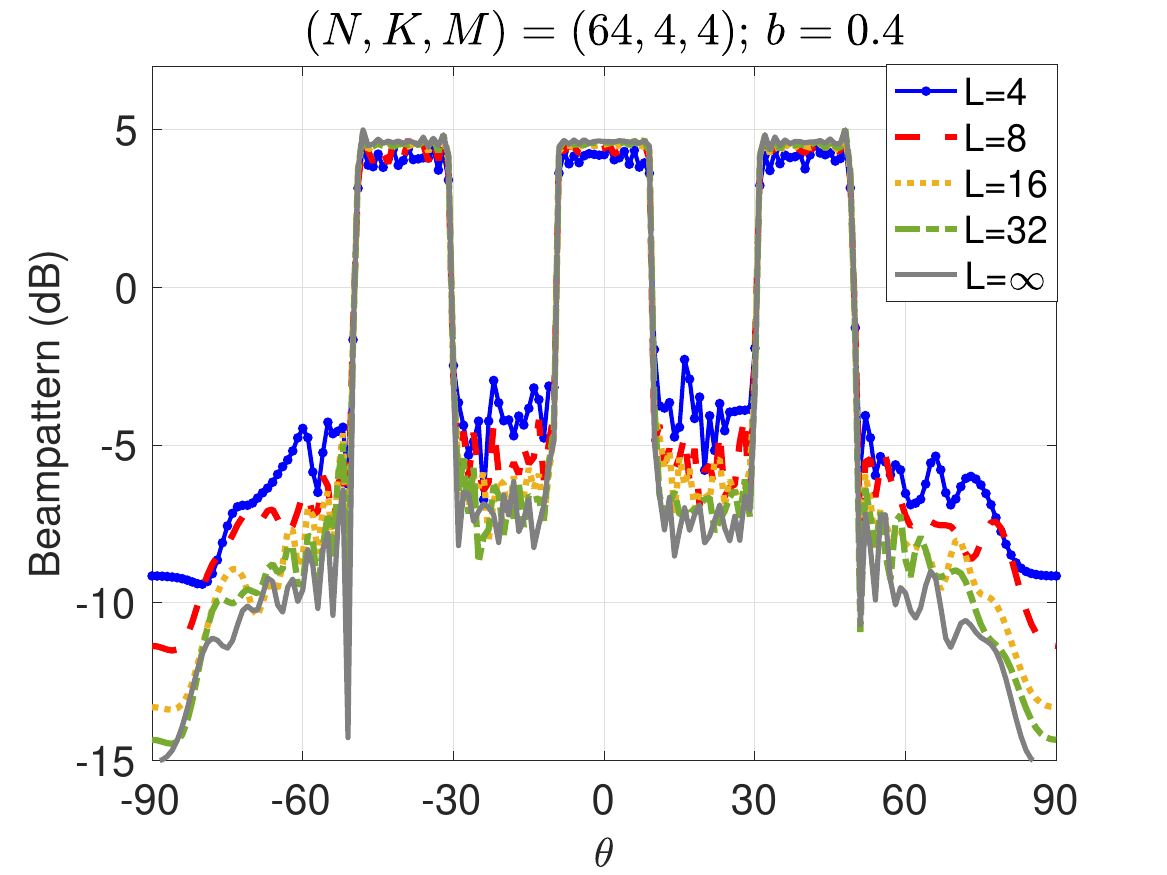}}
\subfigure[Different $b$.]{\includegraphics[width=0.5\columnwidth]{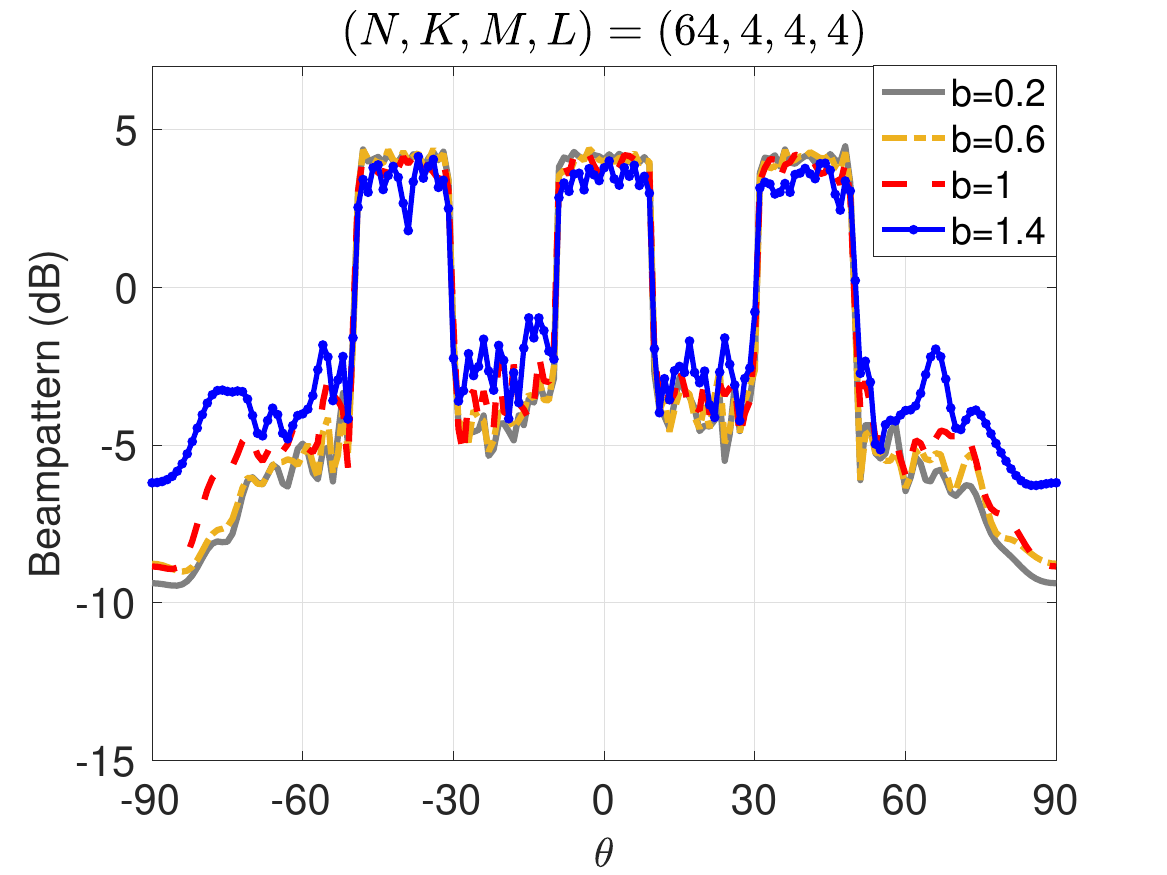}}

\centering
\caption{Achieved beampatterns by the proposed approach under different system configurations.}
\label{pattern}
\end{figure*}
We first test the convergence behaviors of proposed Algorithms \ref{alg:alm} and \ref{alg:BSUM} in Fig. \ref{converge}, where $(N,K,M,L)=(64,4,4,4)$ and $b=0.8$. Specifically, Fig. \ref{converge} (a) shows the  behavior of the inner loop, where proposed BSUM Algorithm \ref{alg:BSUM} is applied to solve the ALM subproblem.  Here,  we say that a successful update is achieved  if  Algorithm \ref{alg:BSUM} reaches an $\epsilon_m$-stationary point of the ALM subproblem within the specified number of iterations. As shown Fig. \ref{converge} (a),  an $\epsilon_m$-stationary point can be obtained within a very small number of inner iterations. Notably, over $80\%$ of the inner loops  terminate in fewer than $10$ iterations and almost all of the inner loops terminates within $20$ iterations. This, combined with the fact that all variables admit closed-form updates, demonstrates the efficiency of proposed BSUM Algorithm \ref{alg:BSUM}. The convergence behavior of the outer loop is presented in Fig. \ref{converge} (b).  As demonstrated, the proposed Algorithm \ref{alg:alm} converges rapidly at the early iterations and successfully meet the stopping criterion in \eqref{stop} within the preset maximum of $500$  iterations.  The good convergence behaviors of both the inner and outer loops make the proposed inexact ALM algorithm  computationally very efficient.

\subsection{Communication and Radar Performance under Different System Configurations}\label{sec:simulation2}
In this subsection, we evaluate the  communication and radar performance achieved by the proposed approach under various system configurations. In particular, we aim at investigating the impact of several problem parameters, including the system parameters $(N,K,M,L)$ and the safety margin threshold $b$, on the communication SER  and the radar beampattern.

We first depict the SER  with respect to the safety margin threshold $b$ under different system configurations in Fig. \ref{SER_system}, where the communication SNR is set as $10$ dB. The theoretical upper and lower bounds on the SEP in \eqref{upperbound} are included as benchmarks. As shown in the figure, the achieved SERs by the proposed approach under all system configurations decrease as $b$ increases and remain much lower than the theoretical bound in most cases. The is  because  some of the  CI constraints in \eqref{linear} are satisfied with strict inequality at the output solution, which makes the achieved SER much better than the theoretical SER derived based on the fixed $b$.  However, when the problem becomes difficult (i.e., when $N/K$ is small, $L$ is small, $M$ is large, and $b$ is large), more CI constraints will be  satisfied with equality or even violated (even the problem itself might be infeasible), leading to a reduction in the gap between the achieved and theoretical SERs, e.g., in the configuration $(N,K,M,L)=(64,4,8,4)$. This also explains why different system configurations exhibit dramatically distinct SER performances with the same $b$. 


In Fig. \ref{pattern}, we depict the achieved beampatterns by the proposed approach under various system configurations and explore the effects of different problem parameters on the achieved beampattern. We first test the impact of the number of transmit antennas on the obtained beampatterns in Fig. \ref{pattern} (a). As demonstrated, the number of transmit antennas has a strong influence on the shape of the obtained beampattern. With more antennas, the obtained beampattern can better approximate the desired square shape given in \eqref{desire}, emitting equal power in the main beam. This demonstrates the advantage of employing massive MIMO technology in the DFRC system. 

In contrast, the number of users $K$,  the quantization level $L$, and the safety margin threshold $b$ have a weak impact on the shape and power level of the main beam; see Fig. \ref{pattern} (b)-(d). Instead, they  affect the power level of undesirable angles. Specifically, the power level of undesirable angles decreases as $K$ and $b$ decrease and $L$ increases ---  that is, when the constraints in the considered problem are easier to be satisfied.  An important engineering observation that can be drawn  from Fig. \ref{pattern} (c) is that  the achieved beampattern with only $L=16$, which corresponds to $3$-bit DACs,  is similar to that achieved by infinite-resolution DACs. Finally,  we note that the constellation order $M$  affects the problem in a similar manner to $b$ (as it only influences the parameter $\bC_t$ involved in the CI constraint). As a result, the impact of $M$ on the achieved beampattern is also similar to that of $b$. As such, we omit the detailed simulation result of $M$. 

\begin{figure}
\includegraphics[scale=0.38]{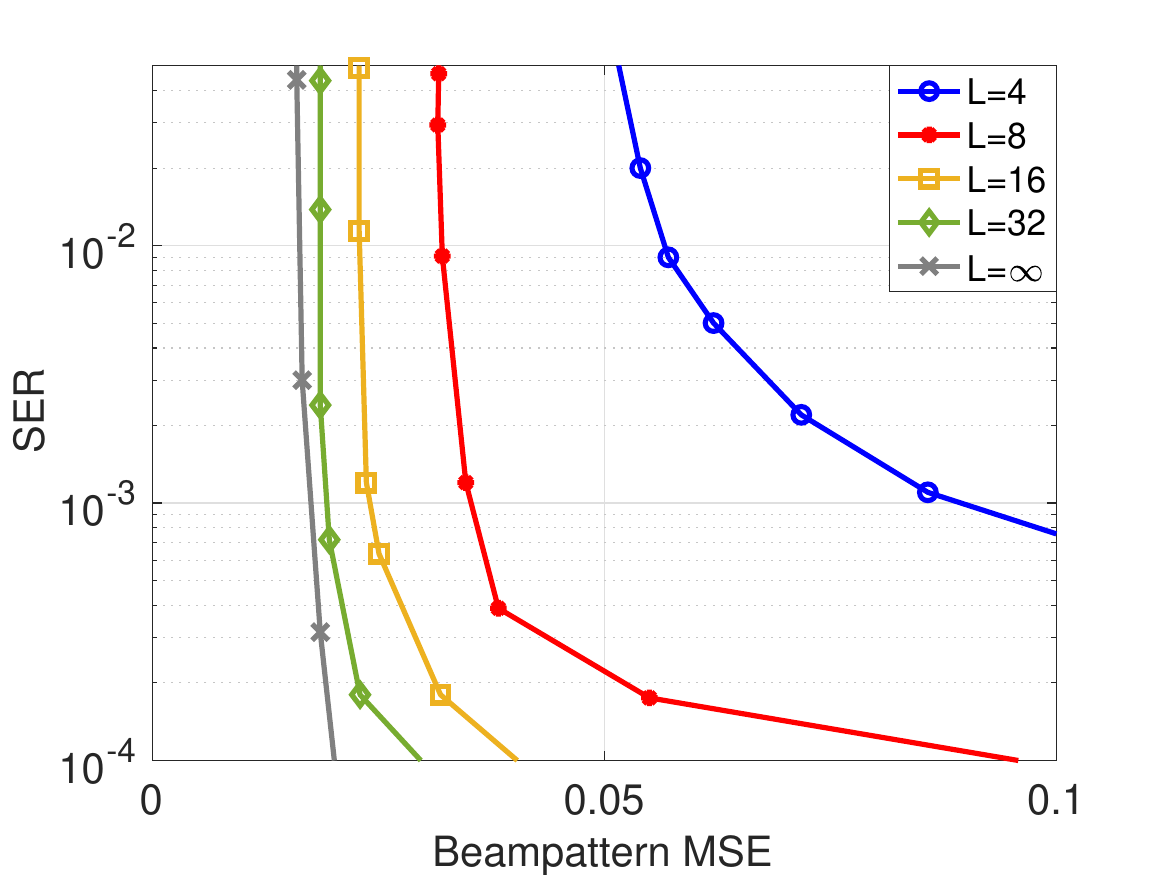}
\centering
\caption{Achieved radar beampattern MSE and communication SER trade-offs achieved by the proposed approach, where $(N,K,M,L)=(64,4,8,4)$ and SNR$=10$ dB.}
\label{MSESER_L}
\vspace{-0.2cm}
\end{figure}
To illustrate the impact of the quantization level $L$ on the system performance more clearly,  we plot the trade-off curves between the radar beampattern MSE and the communication SER for different quantization levels in  Fig. \ref{MSESER_L}. We consider a system with $(N,K,M,L)=(64,4,8,4)$ and set  the SNR as $10$ dB. The curves are obtained by varying the safety margin threshold $b$. It can be observed from the figure that increasing the quantization level from $L=4$ to $L=16$, which corresponds to increasing the resolution of DACs from $1$ bit to $3$ bits, can significantly enhance both the communication and radar performance of the DFRC system.  However, only marginal gains can be obtained if the quantization level is further increased beyond $L=16$, i.e.,  the resolution of DACs is increased beyond 3 bits.


\subsection{Comparison with Existing SOTA Approaches}\label{sec:simulation3}
In this subsection, we compare both the radar and communication performance of the proposed approach with existing SOTA ones.  As no prior work has addressed the general QCE transmit waveform design problem for DFRC systems, we shall focus on the one-bit case,  i.e., $L=4$.  We compare the proposed approach with the algorithm\footnote{In \cite{DFRC2}, several algorithms have been proposed  to deal with different DFRC waveform design problems that are formulated based on different criteria. Here, we  adopt the one designed for solving the weighted radar and communication optimization problem with the total power constraint as the benchmark, as it demonstrates the best radar and communication performance with the presence of one-bit quantization.} in \cite{DFRC2} and the MVAM algorithm in \cite{DFRC1bit1}. Note that the algorithm in \cite{DFRC2} is designed for the infinite-resolution case. 
To obtain one-bit transmit signal,  we directly quantize its output to satisfy the one-bit constraint, which is  termed as ``Quantize''. We also include the unquantized version of the algorithm in \cite{DFRC2} as a benchmark, which is termed as ``Unquantize''.
\begin{figure}
\subfigure[$(N,K,M)=(16,2,4)$.]{\includegraphics[width=0.48\columnwidth]{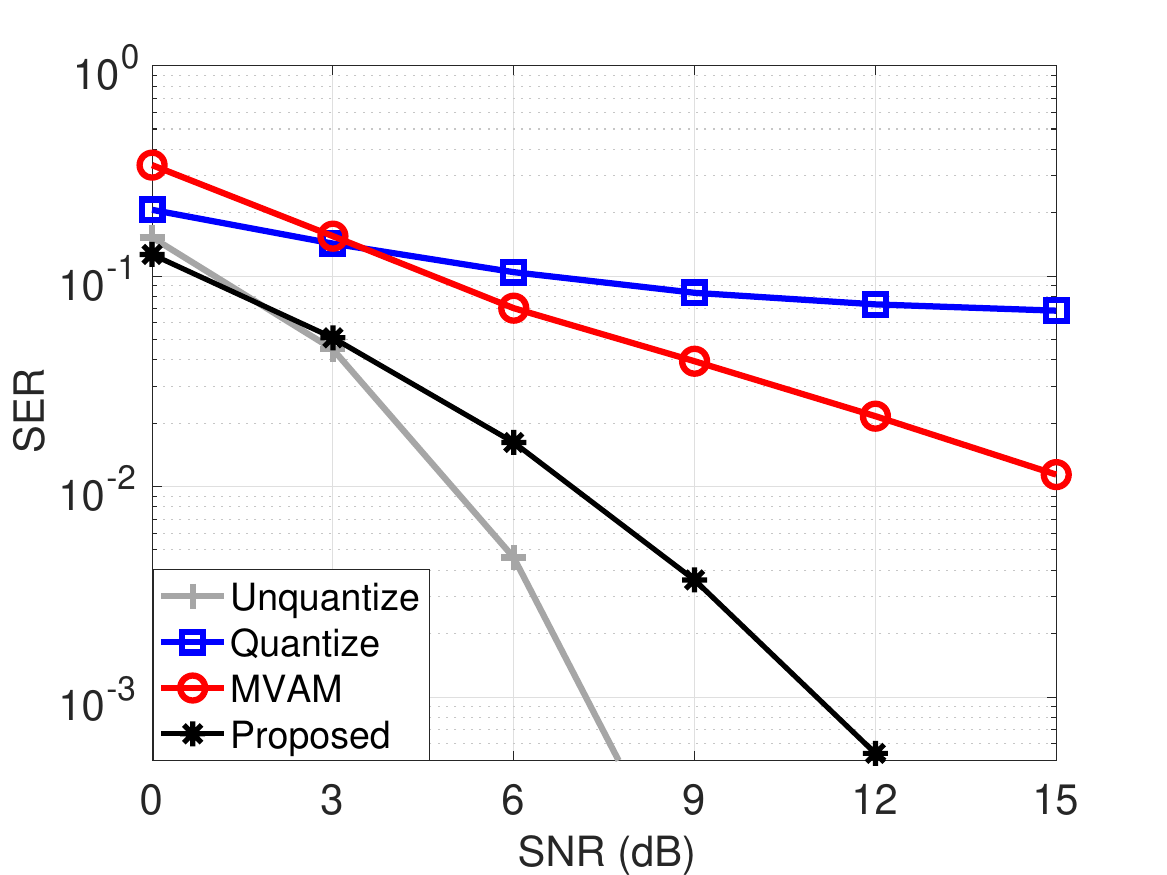}}
\subfigure[$(N,K,M)=(64,4,8)$.]{\includegraphics[width=0.48\columnwidth]{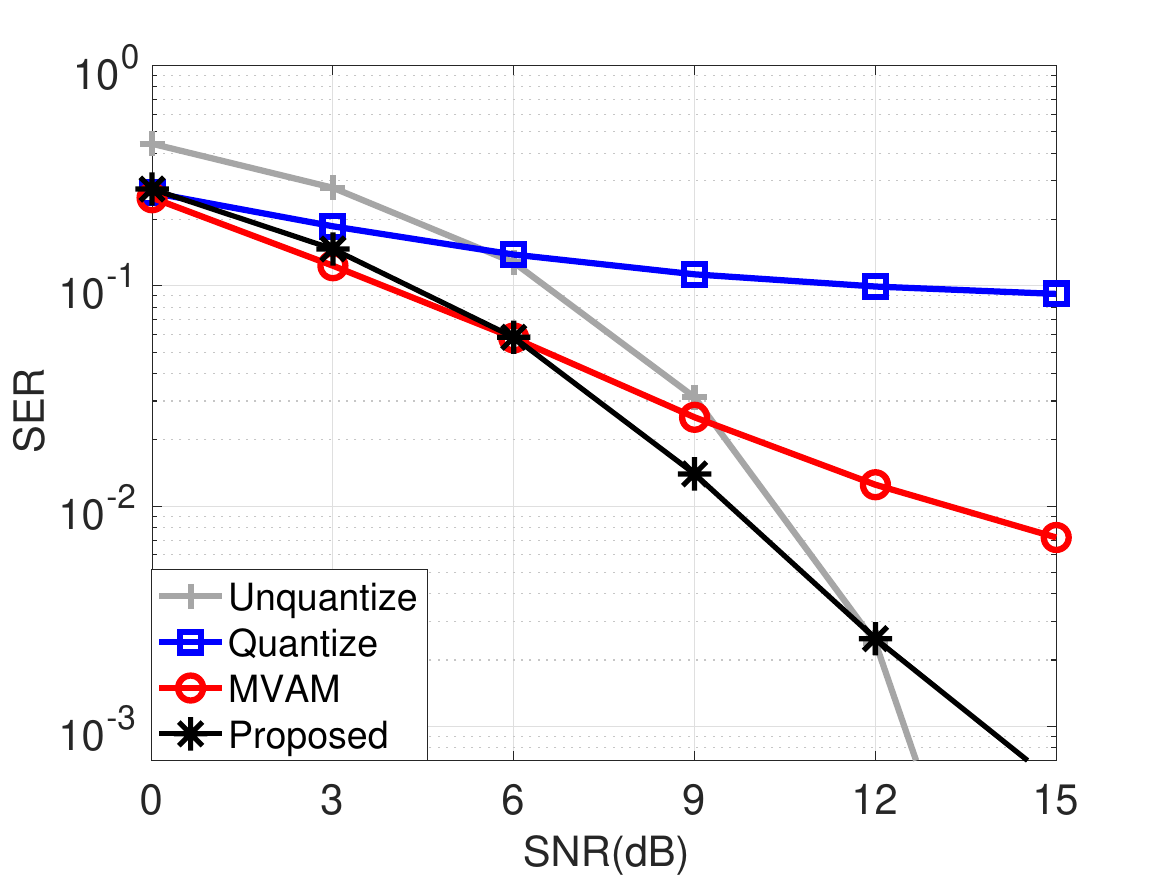}}
\centering
\caption{Achieved SERs by different algorithms, where $b=0.4$.}
\label{SER}
\end{figure}
\begin{figure}
\subfigure[$(N,K,M)=(16,2,4)$.]{\includegraphics[width=0.48\columnwidth]{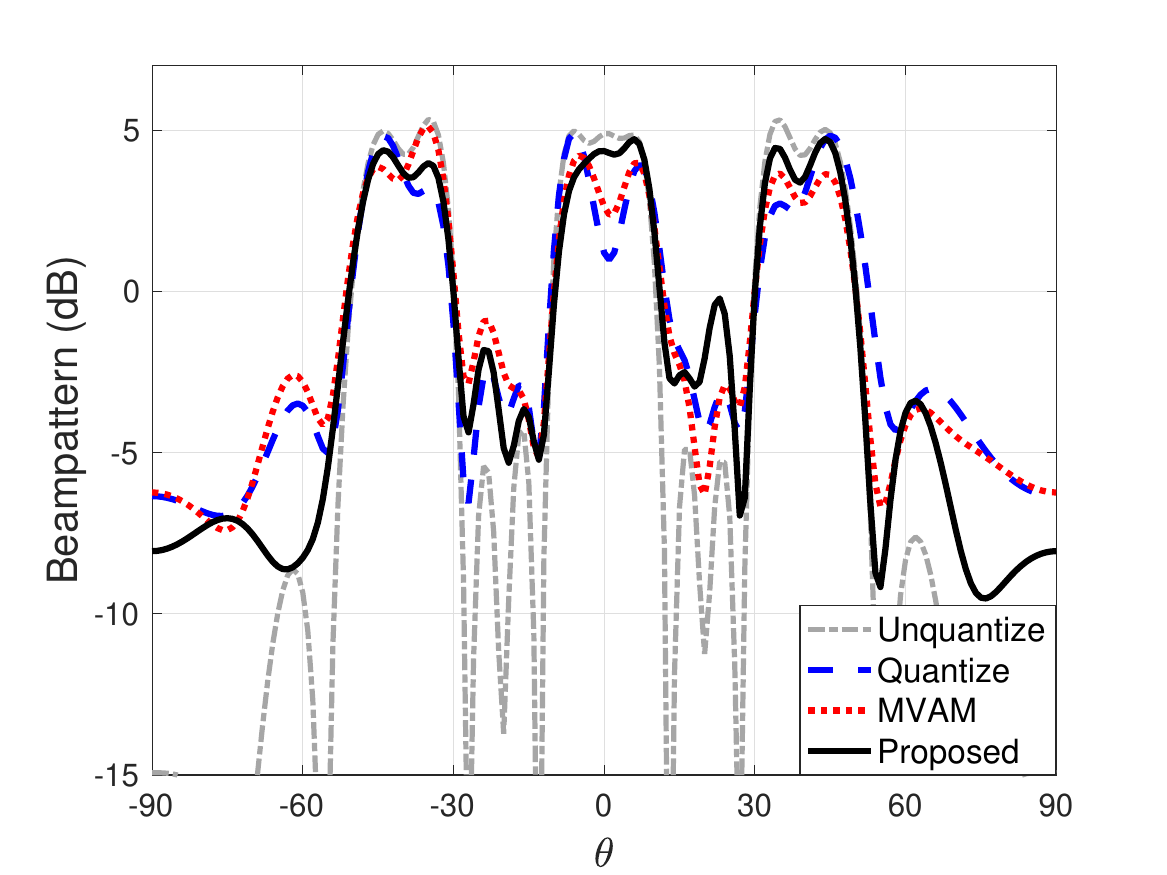}}
\subfigure[$(N,K,M)=(64,4,8)$.]{\includegraphics[width=0.48\columnwidth]{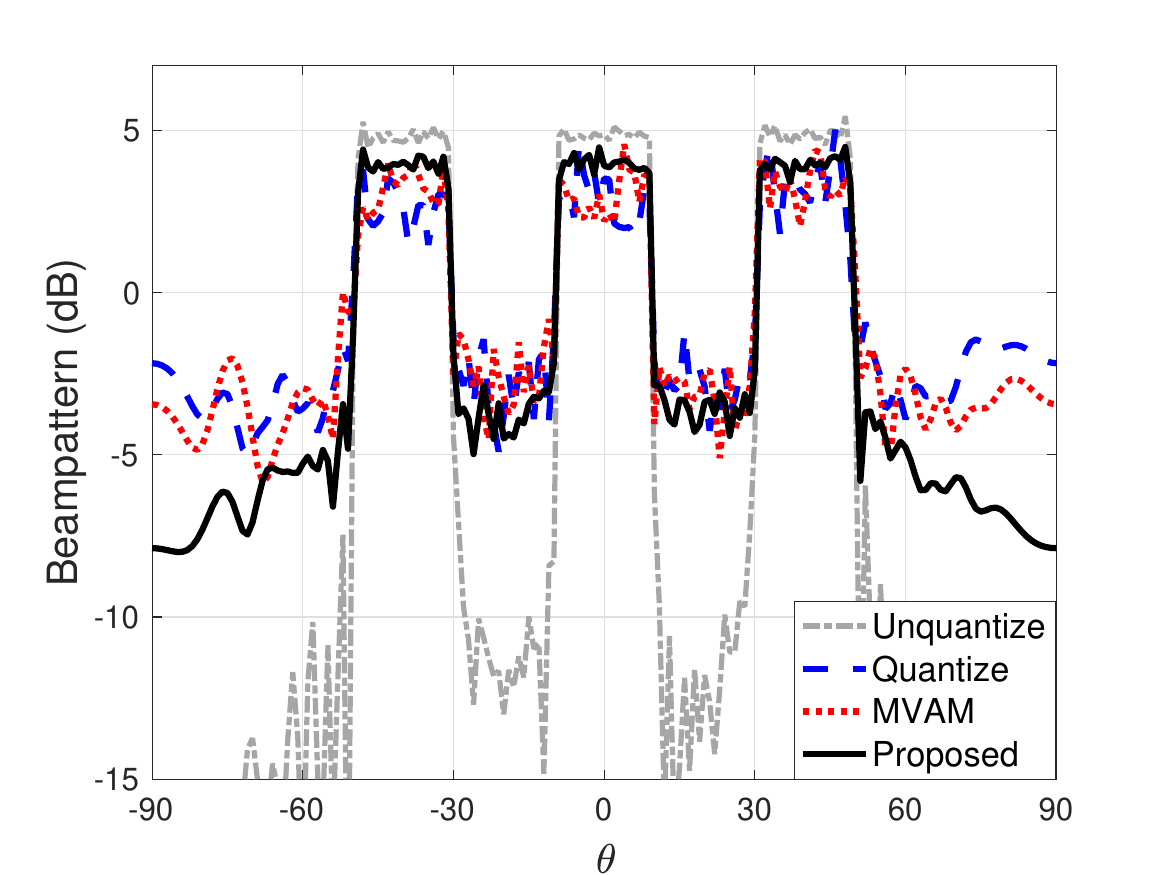}}
\centering
\caption{Achieved beampatterns by different algorithms, where $b=0.4$.}
\vspace{-0.2cm}
\label{beampattern}
\end{figure}

In Figs. \ref{SER} and \ref{beampattern}, we compare the SER and beampattern achieved by the considered four approaches for both a small system with $(N,K,M)=(16,2,4)$ and  a large system with $(N,K,M)=(64,4,8)$.   The safety margin threshold is set as $b=0.4$. Note that both of the algorithms in \cite{DFRC2} and \cite{DFRC1bit1} focus on the optimization of a weighted  radar and communication performance. In our implementation, we  tune the weighting factor to achieve the best balanced radar and communication performance for these algorithms.  As shown in Fig. \ref{SER}, the SER performance of the  algorithm in \cite{DFRC2} suffers from a severe degradation with the direct one-bit quantization.  In comparison, the MVAM algorithm \cite{DFRC1bit1} and the proposed algorithm, which are designed specifically for the one-bit scenario, achieve better SER performance. Due to superior problem formulations, they even outperform the unquantized version of the algorithm in \cite{DFRC2} in large systems and at low SNRs; see Fig. \ref{SER} (b). Of these two approaches, the proposed one exhibits significantly lower SER, especially at high SNRs. Regarding the radar performance, it can be observed from Fig. \ref{beampattern} that the three one-bit approaches, though experience degradation due to the one-bit quantization, still produce satisfactory beampatterns, featuring strong mainlobes around the desired angles. Among them, the proposed approach demonstrates the best beampattern.

To give a more fair comparison, we plot in Fig. \ref{MSEvSER}   the trade-off curves between the radar beampattern MSE and the communication SER of the three one-bit approaches. Specifically, the curves for the algorithms in \cite{DFRC2} and \cite{DFRC1bit1} are obtained by ranging the weighting factors from $0$ to $1$ and the curve for the proposed algorithm is obtained by varying the safety margin threshold $b$. As demonstrated in the figure, the proposed algorithm achieves substantially lower radar beampattern MSE and communication SER  compared to the other two algorithms, especially when the system dimension is large.  In particular, with the same communication SER, the radar beampattern MSE achieved by the proposed approach is generally more than three times lower than those achieved by the SOTA algorithms.   This is attributed to both the superiority of the CI-based problem formulation (compared to the MSE-based formulation)  and the effectiveness of the proposed  approach. 
 \begin{figure}
\subfigure[$(N,K,M)=(16,2,4)$.]{\includegraphics[width=0.48\columnwidth]{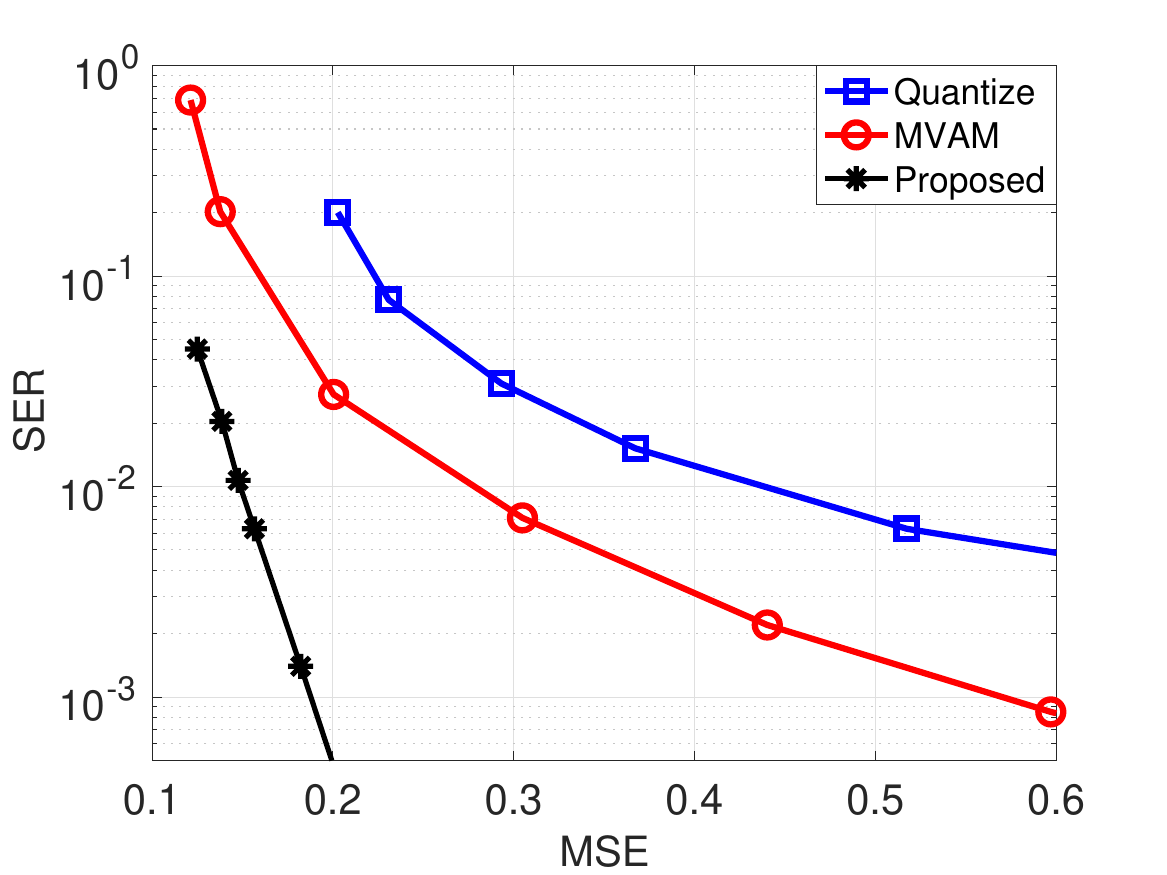}}
\subfigure[$(N,K,M)=(64,4,8)$.]{\includegraphics[width=0.48\columnwidth]{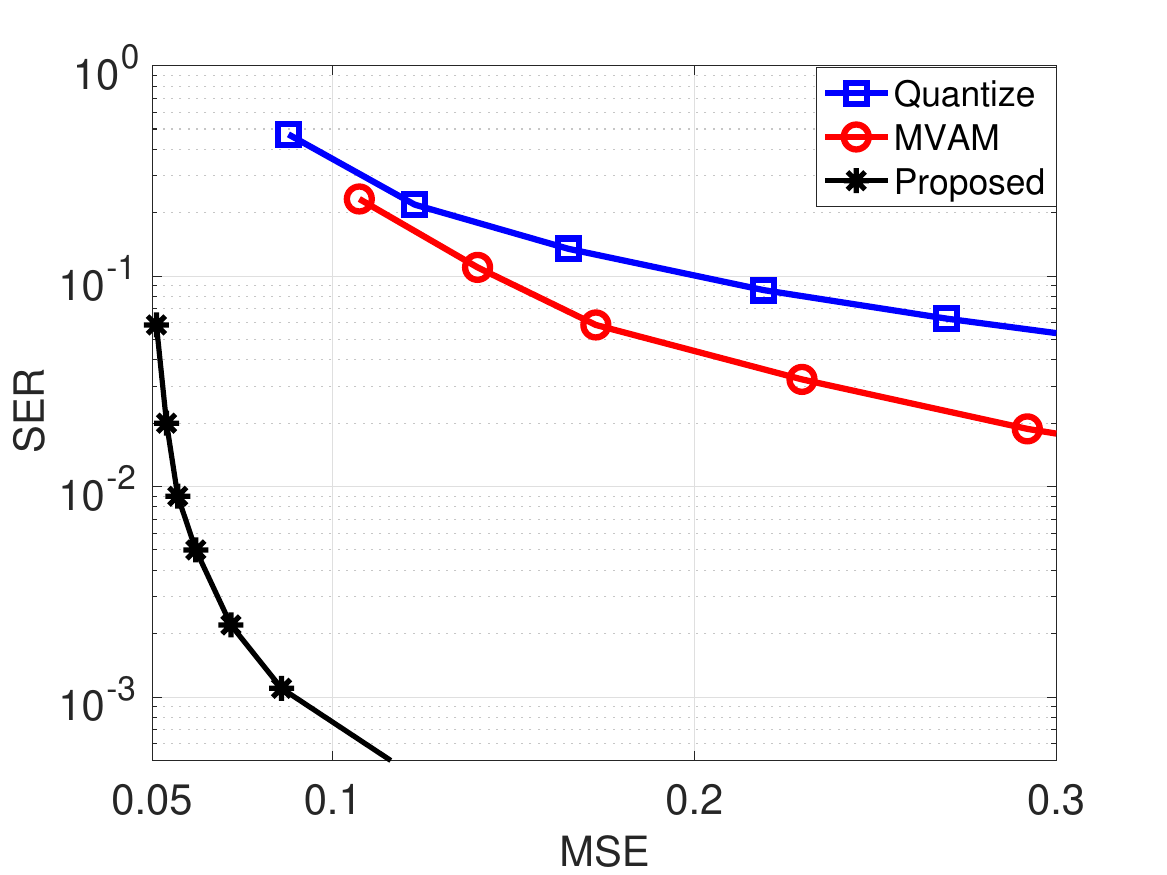}}
\centering
\caption{Achieved radar beampattern  MSE and communication SER trade-offs by different algorithms, where SNR$=10$ dB.}
\label{MSEvSER}
\end{figure}
\section{Conclusion}\label{sec:5}
In this paper, we investigated the general QCE waveform design for massive MIMO DFRC systems. The transmit waveform was optimized to minimize the MSE between the designed and desired beampatterns while satisfying the CI-based communication QoS constraints and the QCE constraint. We proposed an efficient approach for solving  the formulated discrete optimization problem, which consists of two main steps. First, the discrete problem was transformed into an equivalent continuous problem using the penalty technique. Second, an inexact ALM algorithm was developed for solving the nonconvex continuous penalty model, where the ALM subproblem at each iteration was efficiently solved by a custom-designed BSUM algorithm in an inexact fashion. Simulation results demonstrated the good convergence behaviors of the proposed inexact ALM and BSUM algorithms and the superiority of the proposed approach compared to the existing SOTA ones in both radar and communication performance.   An interesting  future work is to optimize the overall radar performance rather than solely focusing on the transmitter design. 

\appendices
\section{Proof of Proposition \ref{the1}}\label{appendixA}
In this appendix, we prove the following proposition, which covers Proposition \ref{the1} as a special case.
\begin{proposition}
Consider the following problem: 
\begin{equation}\label{appen_1}
\begin{aligned}
\min_{\x\in\mathcal{A}}~&F(\x)\\
\text{\normalfont{s.t. }}~&\x_n\in\mathcal{X}_L^\RR,~ \forall~n\in[m],
\end{aligned}
\end{equation}
where $\x\in\mathbb{R}^{2m}$, $F:\mathbb{R}^{2m}\rightarrow \mathbb{R}$ is a continuously differentiable function, $\mathcal{A}\subseteq\mathbb{R}^{2m}$ is a closed set, and $\x_n:=[x_n;x_{m+n}]$. Then there exists $\lambda_0>0$ such that for all $\lambda>\lambda_0$,  problem \eqref{appen_1} shares the same optimal solution with the following penalty model:
\begin{equation}\label{appen_2}
\begin{aligned}
\min_{\x\in\mathcal{A}}~&F(\x)-\lambda\|\x\|^2\\
\text{\normalfont{s.t. }}~&\x_n\in\text{\normalfont{conv}}(\mathcal{X}_L^\RR),~\forall~n\in[m].
\end{aligned}
\end{equation}
\end{proposition}
\begin{proof}
 It suffices to show that for all $\lambda>\lambda_0$, each optimal solution of \eqref{appen_2} is feasible for \eqref{appen_1}. For ease of the notation,  let $$\mathcal{F}_d:=\mathcal{A}\cap\{\x\mid\x_n\in\mathcal{X}_L^\RR,~\forall~n\in[m]\}$$ and $$\mathcal{F}_c:=\mathcal{A}\cap\{\x\mid\x_n\in\text{conv}(\mathcal{X}_L^\RR),~\forall~n\in[m]\}$$ be the feasible sets of the discrete problem \eqref{appen_1} and the continuous penalty problem \eqref{appen_2}, respectively. 
In addition, let $\{{\bv}_1,{\bv}_2,\dots,{\bv}_B\}$ denote all the points in $\mathcal{F}_d$ and define 
 $$\mathcal{S}_1:=\bigcup\limits_{b=1}^B\left\{\x\in\mathcal{F}_c\mid0<\|\x-\bv_b\|<\eta\sin\frac{\pi}{L}\right\}$$
 and  
 $$\mathcal{S}_2:=\bigcup\limits_{b=1}^B\left\{\x\in\mathcal{F}_c\mid\|\x-\bv_b\|\geq\eta\sin\frac{\pi}{L}\right\},$$ where $\mathcal{S}_1$ collects the points in $\mathcal{F}_c$ that are sufficiently close to, but not within, $\mathcal{F}_d$, and $\mathcal{S}_2$ collects the points in $\mathcal{F}_c$ that are away from $\mathcal{F}_d$.  To prove our claim, we need to show that all points in $\mathcal{F}_c\backslash\mathcal{F}_d$ cannot be optimal for problem \eqref{appen_2}. For this purpose,  we next examine the two cases $\x\in\mathcal{S}_1$ and $\x\in\mathcal{S}_2$, separately. 

    \begin{figure}
 \centering
 \includegraphics[scale=0.2]{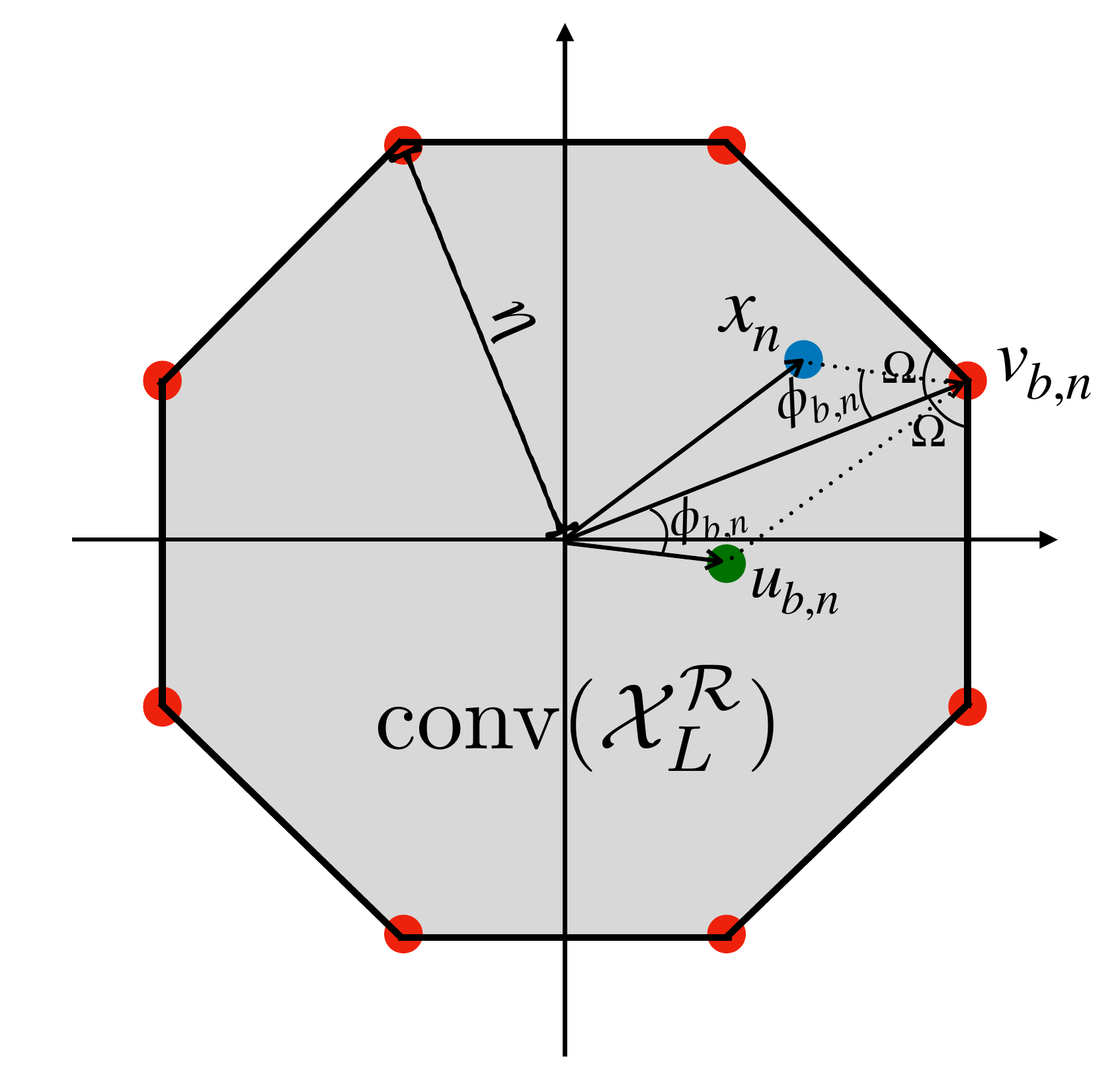}
 \caption{An illustration of $\{\x_n,\bv_{b,n},\bu_{b,n}\}$ and $\phi_{b,n}$, where $\Omega=\frac{\pi}{2}-\frac{\pi}{L}$. Since $\x_n$ lies in $\text{conv}(\mathcal{X}_L^\RR)$, i.e., the shaded area, we have $\phi_{b,n}\leq \Omega$.  The equality holds when $\x_n$ lies on  the two edges  of $\text{conv}(\mathcal{X}_L^\RR)$ that are incident with $\bv_{b,n}$. }
 \vspace{-0.3cm}
 \label{fig:appendixA}
 \end{figure}
\textbf{Case I}: $\x\in \mathcal{S}_1$. For any  $\x\in\mathcal{S}_1$,  there exists $\bv_b\in\mathcal{F}_d$ such that $\|\x-\bv_b\|=\rho<\eta\sin\frac{\pi}{L}$. Let $\bu_b=\bv_b-\bx$ and $\phi_{b,n}$ denote the angle between $\bu_{b,n}$ and $\bv_{b,n}$ for all $n$. Then we have
 \begin{equation}\label{eqn:appendixA1}
 \begin{aligned}
\|\bv_b\|^2-\|\x\|^2&=\|\bv_b\|^2-\|\bv_b-\bu_b\|^2\\&=2\bu_b^\mathsf{T}\bv_b-\|\bu_b\|^2\\&=2\sum_{n=1}^m\bu_{b,n}^\mathsf{T}\bv_{b,n}-\rho^2\\
&\overset{(a)}{=}2\eta\sum_{n=1}^m\|\bu_{b,n}\|\cos\phi_{b,n}-\rho^2\\
&\overset{(b)}{\geq}2\eta\sin\frac{\pi}{L}\rho-\rho^2,
\end{aligned}
\end{equation}
where (a) uses the fact that $\|\bv_{b,n}\|=\eta$, and (b) holds since $\sum_{n=1}^N\|\bu_{b,n}\|\geq\|\bu_b\|=\rho$ and $\phi_{b,n}\leq\frac{\pi}{2}-\frac{\pi}{L}$; see Fig. \ref{fig:appendixA} for an illustration. Since $F(\x)$ is continuously differentiable on the bounded set $\mathcal{F}_c$, it is  Lipschitz continuous on $\mathcal{F}_c$. Let $L_F$ be the Lipschitz constant, then we get 
\begin{equation}\label{eqn:appendixA2}
|F(\bx)-F(\bv_b)|\leq L_F\|\bx-\bv_b\|=L_F\rho.
\end{equation}
Combing \eqref{eqn:appendixA2} with \eqref{eqn:appendixA1} and noting that $\rho<\eta\sin\frac{\pi}{L}$, we have that for any $\lambda>\frac{L_F}{\eta\sin\frac{\pi}{L}}$, 
$$\begin{aligned}
&F(\x)-\lambda\|\x\|^2-(F(\bv_b)-\lambda\|\bv_b\|^2)\\\geq\,&\rho(-L_F+2\lambda\eta\sin\frac{\pi}{L}-\lambda\rho)\\\geq\,&\rho(-L_F+\lambda\eta\sin\frac{\pi}{L})>0,\end{aligned}$$ i.e., $\x$ is not optimal for \eqref{appen_2}.

 \textbf{Case II: $\x\in \mathcal{S}_2$}. Now consider any $\x\in\mathcal{S}_2$. It is easy to check that the set $\mathcal{S}_2$ is compact, and thus 
$$\max_{\x\in\mathcal{S}_2}\|\x\|:=c_0<\eta.$$ 
Let $\bv^*\in\mathcal{F}_d$ and  $F^*:=F(\bv^*)$ be an optimal solution and the optimal value of problem \eqref{appen_1}, respectively, and let  $\hat{F}:=\min_{\x\in\mathcal{S}_2}F(\x).$ Then for all $\x\in\mathcal{S}_2$ and $\lambda>\frac{F^*-\hat{F}}{\eta-c_0}$, we have
$$F(\x)-\lambda \|\x\|^2-(F(\bv^*)-\lambda\|\bv^*\|^2)\geq \hat{F}-F^*+\lambda (\eta^2-c_0^2)>0,$$ where the first inequality holds since $\|\bv^*\|=\eta$. Hence, $\x$ is not optimal for \eqref{appen_2}.

Combining the above two cases, we can now conclude that the desired result holds for  $\lambda_0=\max\{\frac{L_F}{\eta\sin\frac{\pi}{L}},\frac{F^*-\hat{F}}{\eta-c_0}\}.$
 \end{proof}
 \section{Boundedness of the sequences generated by Algorithms \ref{alg:alm} and \ref{alg:BSUM}}\label{appendixB}
In this section, we show that the sequences generated by Algorithms \ref{alg:alm} and \ref{alg:BSUM} are bounded. This result is important to the proof of both Theorem \ref{theorem1} and  Proposition \ref{proposition2}. 
\begin{lemma} \label{bounded}There exists a constant $\overline{M}>0$ (independent of $m$ and $r$) such that
 \begin{enumerate}
 \item The sequence $(\x^m,\bw^m,\bz^m)$ generated by Algorithm \ref{alg:alm} satisfies 
 $$
 \|\x^m\|\leq \overline{M},~\|\bw^m\|\leq\overline{M},~\|\bz^m\|\leq\overline{M},~\forall~m.$$
 \item  For any $m>0$, the sequence $(\x^{(r)},\bw^{(r)},\bz^{(r)})$ generated by Algorithm \ref{alg:BSUM} satisfies
$$\|\x^{(r)}\|\leq \overline{M},~\|\bw^{(r)}\|\leq\overline{M},~\|\bz^{(r)}\|\leq\overline{M},~\forall~r.$$
 \end{enumerate}
\end{lemma}
\begin{proof}
Note that for any $m\geq0$, $(\x^{m+1},\bw^{m+1},\bz^{m+1})$ is the output of Algorithm \ref{alg:BSUM} for solving the ALM subproblem 
\eqref{ALMsubproblem}. Hence, it suffices to prove the second assertion of the  above lemma. 

Given any $m\geq0$ and consider the sequence $\{(\x^{(r+1)},\bw^{(r+1)},\bz^{(r+1)})\}$ generated by Algorithm \ref{alg:BSUM} for solving the ALM subproblem \eqref{ALMsubproblem}. First, since $\x^{(r+1)}\in\mathcal{X}$ and noting the definition of $\mathcal{X}$, we have $$\|\x^{(r+1)}\|\leq \sqrt{NT}\eta:=M_x, ~\forall~r.$$  
In addition, it follows from the update rule of variable $\bz$ in \eqref{zsolution} 
that
$$\begin{aligned}&\|\bz^{(r+1)}\|\\
\leq\,& \|\bC\|\|\x^{(r+1)}\|+\|\bb\|+\frac{\|\bmu^m\|}{\rho_\mu^m}\\
\leq\,& \|\bC\|M_x\hspace{-0.05cm}+\hspace{-0.05cm}\|\bb\|\hspace{-0.05cm}+\hspace{-0.05cm}\frac{\sqrt{2KT}\max\{|\mu_{\min}|,|\mu_{\max}|\}}{\rho_\mu^0}:=M_z,~\forall~r,
\end{aligned}$$
where the second inequality holds since the elements of $\bmu$ are lower and upper bounded by $\mu_{\min}$ and $\mu_{\max},$ respectively, and $\rho_\mu^m$ is increasing in $m$ and thus lower bounded by $\rho_\mu^0$; see \eqref{multiplier} and \eqref{rho}. The remaining task is to show that $\{\bw^{(r)}\}$ is bounded by a constant (independent of $m$ and $r$). According to the update rule of $\bw^{(r)}$ in \eqref{wsolution}, we only need to prove that  $\{\beta_q^{(r+1)}\}$ is bounded for all $q\in[Q]$. 

Recall that  $\beta_q^{(r+1)}$ is the unique positive root of  the cubic equation:
$$ P(\beta):=4\beta^3+\rho_v^m\beta-\left\|\boldsymbol{\xi}_{{q}}^{(r)}\right\|=0,$$
where\footnote{Here and after, for a vector $\boldsymbol{\xi}$, we use $\boldsymbol{\xi}_{(q)}$ to denote its subvector that collects the elements in $\boldsymbol{\xi}$ at the same positions as $\bw_{(q)}$ in $\bw$.} $\boldsymbol{\xi}_{{q}}^{(r)}=\left(\nabla g(\bw^{(r)})-\bnu^m-\rho_\nu^m\bA\x^{(r+1)}\right)_{(q)}$.
Using the the definitions of $g(\bw)$ and $\beta_q^{(r)}$, we can further express $\boldsymbol{\xi}_{q}^{(r)}$ as 
$$\boldsymbol{\xi}_{{q}}^{(r)}=\left(4\sum_{q\in[Q]}c_q(\beta_q^{(r)})^2\right)c_q\bw_{(q)}^{(r)}-\bnu^m_{(q)}-\rho_\nu^m\left(\bA\x^{(r+1)}\right)_{(q)}.$$
It is obvious that $P(\beta)$ is increasing in $\beta$ with $P(0)<0$ and $P(+\infty)=+\infty$. Hence, if we can find a positive constant $c$ such that $P(c)>0$,  then $\beta_q^{(r+1)}\leq c$. Note that $P(\beta)$ can be lower bounded as 
$$
\begin{aligned}
P(\beta)\geq&4\beta^3+\rho_\nu^m\beta-4\left(\sum_{q=1}^Qc_q(\beta_q^{(r)})^2\right)c_q\beta_q^{(r)}\\
&-\left\|\boldsymbol{\nu}^m_{(q)}\right\|-\rho_\nu^m\left\|(\bA\x^{(r+1)})_{(q)}\right\|,
\end{aligned}$$
where we have used the fact that $\|\bw_{(q)}^{(r)}\|=\beta_q^{(r)}$. For the last two terms in the above lower bound, there exists a constant $M_{w,0}>0$ independent of $m$ and $r$ such that 
$$\left\|\boldsymbol{\nu}_{(q)}^m\right\|+\left\|\rho_{\nu}^m(\bA\x^{(r+1})_{(q)}\right\|\leq \rho_\nu^m M_{w,0}.$$
Let $\mathcal{I}=\{q\in[Q]\mid c_q>0\}$, where $c_q$ is given in \eqref{def:cq}. We next investigate the two cases $q\notin\mathcal{I}$ and $q\in\mathcal{I}$ separately. 

In case of $q\notin\mathcal{I}$, i.e., $c_q=0$, we have
$$P(\beta)\geq 4\beta^3+\rho_v^m(\beta-M_{w,0}),$$ which further implies that 
$P(M_{w,0})>0$ and $\beta_q^{(r+1)}\leq M_{w,0}$.

In case of $q\in\mathcal{I}$,  it is straightforward to verify that $P(\beta)>0$ if  $$\beta>\max\left\{\left(\sum_{q\in\mathcal{I}_q}c_q(\beta_q^{(r)})^2\right)^{\frac{1}{3}}(c_q\beta_q^{(r)})^{\frac{1}{3}},M_{w,0}\right\}.$$ Then we immediately get
\begin{equation*}\label{rqbound}
\beta_q^{(r+1)}\leq \max\left\{\left(\sum_{q\in\mathcal{I}_q}c_q(\beta_q^{(r)})^2\right)^{\frac{1}{3}}(c_q\beta_q^{(r)})^{\frac{1}{3}},M_{w,0}\right\}.
\end{equation*}
Squaring the above inequality,  multiplying both sides by $c_q$, and then summing over $q\in\mathcal{I}$, we obtain 
\begin{equation}\label{rqbound}
\begin{aligned}
&\sum_{q\in\mathcal{I}}c_q(\beta_q^{(r+1)})^2\\\leq& \max\hspace{-0.05cm}\left\{\hspace{-0.05cm}\left(\sum_{q\in\mathcal{I}}c_q(\beta_q^{(r)})^2\hspace{-0.05cm}\right)^{\hspace{-0.13cm}\frac{2}{3}}\hspace{-0.1cm}\left(\sum_{q\in\mathcal{I}}c_q(c_q\beta_q^{(r)})^{\frac{2}{3}}\hspace{-0.05cm}\right)\hspace{-0.1cm},\,\sum_{q\in\mathcal{I}}c_qM_{w,0}^2\hspace{-0.05cm}\right\}\vspace{-0.05cm}.
\end{aligned}\end{equation}
In addition, using $\sum_{q\in\mathcal{I}}c_q^2=1$ (see \eqref{def:cq}) and   the (extended) Holder inequality, i.e.,

$$\left(\sum_{q\in\mathcal{I}}|u_q|^a|v_q|^b\right)^{a+b}\hspace{-0.1cm}\leq\left(\sum_{q\in\mathcal{I}}|u_q|^{a+b}\right)^a\left(\sum_{q\in\mathcal{I}}|v_q|^{a+b}\right)^b\hspace{-0.05cm},$$
with 
$u_q=c_q^2,~v_q=c_q(\beta_q^{(r)})^2, a=\frac{2}{3}, $ and $b=\frac{1}{3}$, we have
$$\begin{aligned}
\sum_{q\in\mathcal{I}}c_q(c_q\beta_q^{(r)})^{\frac{2}{3}}&\leq\left(\sum_{q\in\mathcal{I}}c_q(\beta_q^{(r)})^2\right)^{\hspace{-0.1cm}\frac{1}{3}}\hspace{-0.1cm}.
\end{aligned}$$
 Plugging the above inequality into \eqref{rqbound} yields $$\sum_{q\in\mathcal{I}}c_q(\beta_q^{(r+1)})^2\leq \max\left\{\sum_{q\in\mathcal{I}}c_q(\beta_q^{(r)})^2,\sum_{q\in\mathcal{I}}c_qM_{w,0}^2\right\},$$
which further implies
\begin{equation}\label{decrease_w}
\sum_{q\in\mathcal{I}}c_q(\beta_q^{(r+1)})^2\leq \max\left\{\sum_{q\in\mathcal{I}}c_q(\beta_q^{(0)})^2,\sum_{q\in\mathcal{I}}c_qM_{w,0}^2\right\},~\forall ~r.
\end{equation}
Note that $\beta_q^{(0)}=\|\bw_{(q)}^m\|$ (since the  initial point of Algorithm \ref{alg:BSUM} is set as $(\x^{(0)},\bw^{(0)},\bz^{(0)})=(\x^m,\bw^m,\bz^m)$) and $(\bx^{m+1},\bw^{m+1},\bz^{m+1})$ is the output of Algorithm \ref{alg:BSUM}. Then, it follows from \eqref{decrease_w} that 
\begin{equation}\label{decrease_w2}
\sum_{q\in\mathcal{I}}c_q\|\bw_{(q)}^{m+1}\|^2\leq\max\left\{\sum_{q\in\mathcal{I}}c_q\|\bw_{(q)}^m\|^2,\sum_{q\in\mathcal{I}}c_qM_{w,0}^2\right\},~\forall ~m.
\end{equation}
 Combining \eqref{decrease_w2} with \eqref{decrease_w}, we get
\begin{equation*}
\begin{aligned}
\sum_{q\in\mathcal{I}}c_q(\beta_q^{(r+1)})^2\overset{(a)}{\leq}&\max\left\{\sum_{q\in\mathcal{I}}c_q\|\bw_{(q)}^m\|^2,\sum_{q\in\mathcal{I}}c_qM_{w,0}^2\right\}\\\overset{(b)}{\leq} &\max\left\{\sum_{q\in\mathcal{I}}c_q\|\bw_{(q)}^0\|^2,\sum_{q\in\mathcal{I}}c_qM_{w,0}^2\right\},~\forall ~r,
\end{aligned}
\end{equation*}
where (a) is obtained by replacing $\beta^{(0)}_q$ in \eqref{decrease_w} with $\|\bw_{(q)}^{m}\|$ and (b) is obtained by iteratively applying  \eqref{decrease_w2} for $m-1, \dots, 0$. 
Hence, since $c_q>0$ for $q\notin\mathcal{I}$, we further have
$$
\begin{aligned}
&\sum_{q\in\mathcal{I}}\left(\beta_q^{(r+1)}\right)^2\\\leq \,&\frac{\max\left\{\sum_{q\in\mathcal{I}}c_q\|\bw_{(q)}^0\|^2,\sum_{q\in\mathcal{I}}c_qM_{w,0}^2\right\}}{\min_{q\in\mathcal{I}}c_q}:=M_{w,1},~\forall~r.
\end{aligned}$$
Let $M_w:=\sqrt{(Q-|\mathcal{I}|)M_{w,0}^2+M_{w,1}}$, where $|\mathcal{I}|$ denotes the number of elements in $\mathcal{I}$. Then it is simple to check that 
$\|\bw^{(r)}\|\leq M_w$ for all $r$.

Combing the above discussions,  we can now conclude that 
$\overline{M}:=\max\{M_x, M_z, M_w\}$ satisfies the condition in Lemma \ref{bounded}, 
 which completes the proof.
\end{proof}

 \section{Proof of Theorem \ref{theorem1}}\label{appendixC}
According to Lemma \ref{bounded},  the sequence $\{(\x^{m+1},\bw^{m+1},\bz^{m+1})\}$ generated by Algorithm \ref{alg:alm} is bounded. 
As a result, to prove Theorem \ref{theorem1}, we only need to prove the first assertion in it (and the second assertion follows directly from the first assertion and the fact that $\{(\x^{m+1},\bw^{m+1},\bz^{m+1})\}$ lies in a compact set). The following  proof follows a similar argument to that in \cite[Theorems 3.3-3.4]{boundALM}.

Let $(\bar{\x},\bar{\bw},\bar{\bz})$ be any limit point of $\{(\x^{m+1},\bw^{m+1},\bz^{m+1})\}$. We first show that $(\bar{\x},\bar{\bw},\bar{\bz})$ is feasible, i.e., $(\bar{\x},\bar{\bw},\bar{\bz})\in\mathcal{F}$. Since $\mathcal{X}$ and $\mathcal{Z}$ are closed sets and $\bx^{m+1}\in\mathcal{X}$ and $\bz^{m+1}\in\mathcal{Z}$ for all $m$, we have $\bar{\x}\in\mathcal{X}$ and $\bar{\bz}\in\mathcal{Z}$. The remaining task is to show that $(\bar{\x},\bar{\bw},\bar{\bz})$ satisfies the linear constraints in $\mathcal{F}$, i.e., $\bC\bar{\x}-\bar{\bz}=\bb$  and $\bA\bar{\x}=\bar{\bw}$. From the update rule in \eqref{rho}, we can express $(\rho_\mu^m, \rho_\nu^m)=(a_\mu\rho^m, a_\nu\rho^m)$, where $a_\mu>0$, $a_\nu>0$, and $\rho^m$ is increasing in $m$. Hence, either $\rho^m\to\infty$ or $\rho^m$ keeps fixed after a finite number of iterations. In the latter case,  the violation of the linear constraints will be decreased by a factor of $\delta<1$ at each iteration and thus tends to zero as $m\to\infty$, which further implies that  $(\bar{\x},\bar{\bw},\bar{\bz})\in\mathcal{F}$.

 We next consider the case where  $\rho^m\to\infty$. Assume, without loss of generality, that $(\x^{m+1},\bw^{m+1},\bz^{m+1})\rightarrow (\bar{\x},\bar{\bw},\bar{\bz})$. Since  $(\x^{m+1},\bw^{m+1},\bz^{m+1})$ is an $\epsilon_m$-stationary point of the ALM subproblem, there exists 
$\mathbf{e}^m:=\left[{\mathbf{e}^m_x}^\mathsf{T},{\mathbf{e}^m_y}^\mathsf{T},{\mathbf{e}^m_z}^\mathsf{T}\right]^\mathsf{T}$ such that 
\begin{equation}\label{em2}
\begin{aligned}\mathbf{e}^{m}\hspace{-0.05cm}\in\hspace{-0.05cm}\nabla& \mathcal{L}_m(\bx^{m+1},\bw^{m+1}\,\bz^{m+1})\\&\qquad\qquad+\partial \mathbb{I}_{\mathcal{X}}({\x}^{m+1})+\partial\mathbb{I}_{\mathcal{Z}}({\bz}^{m+1})
\end{aligned}\end{equation}with $\|\be^m\|\leq\epsilon_m$ (see \eqref{inexactcondition}).
Substituting $\mathcal{L}_m(\x^{m+1},\bw^{m+1},\bz^{m+1})$ in \eqref{def:lm} into \eqref{em2}, we get
\begin{subequations}\label{the1:eqn1}
\begin{align}
&\hspace{-0.25cm}\left<\nabla h(\x^{m+1})\hspace{-0.05cm}+\hspace{-0.05cm}\bC^\mathsf{T}\bmu^m\hspace{-0.05cm}+\hspace{-0.05cm}\bA^\mathsf{T}\bnu^m\hspace{-0.05cm}+\hspace{-0.05cm}\rho_\mu^m\bC^\mathsf{T}(\bC\x^{m+1}\hspace{-0.05cm}-\hspace{-0.05cm}\bz^{m+1}\hspace{-0.05cm}-\hspace{-0.05cm}\bb)\right.\notag\\
&\hspace{-0.25cm}~\left.+\rho_\nu^m\bA^\mathsf{T}(\bA\x^{m+1}\hspace{-0.05cm}-\hspace{-0.05cm}\bw^{m+1})\hspace{-0.05cm}-\hspace{-0.05cm}\be_x^{m},\x-\x^{m+1}\right>\geq 0,~\forall~\x\in\mathcal{X},\label{the1eqn1:1}\\
&\hspace{-0.25cm}\nabla f(\bw^{m+1})+\nabla g(\bw^{m+1})-\bnu^m\notag\\
&\hspace{2.4cm}-\rho_\nu^m(\bA\x^{m+1}-\bw^{m+1})-\be_w^{m}=\mathbf{0},\label{the1eqn1:2}\\
 &\hspace{-0.25cm}\left<-\bmu^m-\rho_\mu^m(\bC\x^{m+1}-\bz^{m+1}-\bb)-\mathbf{e}_z^m, \bz-\bz^{m+1}\right>\geq0,\notag\\&\hspace{6.4cm}\forall~\bz\in\mathcal{Z},\label{the1eqn1:3}
\end{align}
\end{subequations}
where we have used the fact that for a convex set $\mathcal{X}$, $\boldsymbol{\eta}\in\partial\mathbb{I}_{\mathcal{X}}(\tilde{\x})$ is equivalent to $\boldsymbol{\eta}^\mathsf{T}(\x-\tilde{\x})\leq 0$ for all $\x\in\mathcal{X}.$
By further dividing \eqref{the1eqn1:1}--\eqref{the1eqn1:3} by $\rho^m$, we obtain
\begin{subequations}\label{the1:eqn2}
\begin{align}
&\hspace{-0.25cm}\left<\hspace{-0.03cm}\frac{\nabla h(\x^{m+1})\hspace{-0.05cm}+\hspace{-0.05cm}\bC^\mathsf{T}\bmu^m\hspace{-0.05cm}+\hspace{-0.05cm}\bA^\mathsf{\hspace{-0.05cm}T}\bnu^m\hspace{-0.05cm}-\hspace{-0.05cm}\be_x^m}{\rho^m}\hspace{-0.05cm}+\hspace{-0.05cm}a_\mu\bC^\mathsf{T}(\bC\x^{m+1}\hspace{-0.05cm}-\hspace{-0.05cm}\bz^{m+1}\hspace{-0.05cm}\right.\notag\\
&\left.{\color{white}\frac{x^m}{\rho^m}}\hspace{-0.5cm}-\hspace{-0.05cm}\bb)\hspace{-0.05cm}+\hspace{-0.05cm}a_\nu\bA^\mathsf{\hspace{-0.05cm}T}(\bA\x^{m+1}\hspace{-0.05cm}-\hspace{-0.05cm}\bw^{m+1}),\,\x-\x^{m+1}\hspace{-0.07cm}\right>\geq 0,~\hspace{-0.05cm}\forall~\hspace{-0.05cm}\x\in\mathcal{X}\hspace{-0.05cm},\label{the1eqn2:1}\\
&\hspace{-0.25cm}\frac{\nabla f(\bw^{m+1})+\nabla g(\bw^{m+1})-\bnu^m-\be_w^{m}}{\rho^m}\notag\\
&\hspace{3cm}-a_\nu(\bA\x^{m+1}-\bw^{m+1})=\mathbf{0},\label{the1eqn2:2}\\
 &\hspace{-0.25cm}\left<-\frac{\bmu^m-\mathbf{e}_z^m}{\rho^m}-a_\mu(\bC\x^{m+1}-\bz^{m+1}-\bb), \bz-\bz^{m+1}\right>\geq0,\notag\\&\hspace{6.4cm}\forall~\bz\in\mathcal{Z}.\label{the1eqn2:3}
\end{align}
\end{subequations}
Recall that $(\x^{m+1},\bw^{m+1},\bz^{m+1})$ lies in a compact set. Since $\nabla h(\x), \nabla f(\bw)$, and $\nabla g(\bw)$ are all continuous, they are bounded on that compact set. In addition, the multipliers $(\bmu^m,\bnu^m)$ are bounded and $\|\be^m\|\leq \epsilon_m$ with $\epsilon_m\to0$. Therefore, by letting $m\to\infty$ in  \eqref{the1eqn2:1}--\eqref{the1eqn2:3}, we have 
\begin{subequations}
\begin{align}
&\hspace{-0.15cm}\langle a_\mu\bC^\mathsf{T}(\bC\bar{\x}\hspace{-0.05cm}-\hspace{-0.05cm}\bar{\bz}\hspace{-0.05cm}-\hspace{-0.05cm}\bb)\hspace{-0.05cm}+\hspace{-0.05cm}a_\nu\bA^\mathsf{T}(\bA\bar{\x}\hspace{-0.05cm}-\hspace{-0.05cm}\bar{\bw}),\x\hspace{-0.05cm}-\hspace{-0.05cm}\bar{\x}\rangle\geq 0,~\forall~\hspace{-0.05cm}\x\in\mathcal{X},\label{the1eqn3:1}\\
&\hspace{-0.05cm}a_\nu(\bar{\bw}-\bA\bar{\x})=\mathbf{0},\label{the1eqn3:2}\\
&\hspace{-0.15cm}\langle a_\mu(\bar{\bz}-\bC\bar{\x}+\bb),\bz-\bar{\bz}\rangle\geq 0~,\forall~\bz\in\mathcal{Z}.\label{the1eqn3:3}
\end{align}
\end{subequations}
From \eqref{the1eqn3:2}, $\bar{\bw}-\bA\bar{\x}=\mathbf{0}$. Substituting this into \eqref{the1eqn3:1} and combining  \eqref{the1eqn3:1} with  \eqref{the1eqn3:3}, we get 
$$\left<\bC\bar{\x}-\bar{\bz}-\bb,(\bC\x-\bz)-(\bC\bar{\x}-\bar{\bz})\right>\geq0,~\forall~\x\in\mathcal{X},~\bz\in\mathcal{Z},$$
which further implies that 
$$\begin{aligned}
\hspace{-0.03cm}\|\bC\bar{\x}\hspace{-0.05cm}-\hspace{-0.05cm}\bar{\bz}\hspace{-0.05cm}-\hspace{-0.05cm}\bb\|^{\hspace{-0.01cm}2}\hspace{-0.03cm}&\hspace{-0.02cm}\leq\hspace{-0.05cm} \left<\bC\bar{\x}-\bar{\bz}-\bb,\bC\x-\bz-\bb\right>\\
&\hspace{-0.05cm}\leq\hspace{-0.1cm}\|\bC{\x}\hspace{-0.05cm}-\hspace{-0.05cm}{\bz}\hspace{-0.05cm}-\hspace{-0.05cm}\bb\|\|\bC\bar{\x}\hspace{-0.05cm}-\hspace{-0.05cm}\bar{\bz}-\bb\|,\,\forall\hspace{0.03cm}\x\in\mathcal{X},\hspace{0.03cm}\bz\in\mathcal{Z},
\end{aligned}$$
i.e., 
$\|\bC\bar{\x}-\bar{\bz}-\bb\|\leq \|\bC{\x}-{\bz}-\bb\|$ for all $\x\in\mathcal{X}$ and $\bz\in\mathcal{Z}.$
Hence, $\bC\bar{\x}-\bar{\bz}-\bb=\mathbf{0}$ (as long as $\mathcal{F}$ is nonempty), and $(\bar{\bx},\bar{\bw},\bar{\bz})\in\mathcal{F}$. 

Finally, we show that $(\bar{\x},\bar{\bw},\bar{\bz})$ satisfies \eqref{optcondition}. Given any $(\bx,\bw,\bz)\in\mathcal{F}$,  by taking an inner product of \eqref{the1eqn1:2} with $\bw-\bw^{m+1}$ and then summing it  with \eqref{the1eqn1:1} and \eqref{the1eqn1:3}, we obtain
 \begin{equation}\label{convergence3}
\begin{aligned}
&\langle\nabla f(\bw^{m+1})+\nabla g(\bw^{m+1}),{\bw}-\bw^{m+1}\rangle\\
&\hspace{4.4cm}+\langle \nabla h(\x^{m+1}),{\x}-\x^{m+1}\rangle\\
\overset{(a)}{\geq} &\langle \bmu^m, \bC\x^{m+1}-\bz^{m+1}-\bb\rangle+\langle\bnu^m,\bA\x^{m+1}-\bw^{m+1}\rangle\\
&+\rho_\mu^{m}\,\|\bC\x^{m+1}-\bz^{m+1}-\bb\|^2+\rho_\nu^{m}\,\|\bA\x^{m+1}-\bw^{m+1}\|^2\\
&+\langle\mathbf{e}_x^{m},\x-\x^{m+1}\rangle+\langle\mathbf{e}_w^{m},{\bw}-\bw^{m+1}\rangle+\langle\mathbf{e}_z^{m},{\bz}-\bz^{m+1}\rangle\\
\geq&\hspace{-0.05cm}-\hspace{-0.05cm}\|\bmu^m\|\|\bC\x^{m+1}\hspace{-0.05cm}-\hspace{-0.05cm}\bz^{m+1}-\bb\|\hspace{-0.08cm}-\hspace{-0.05cm}\|\bnu^m\|\|\bA\x^{m+1}\hspace{-0.08cm}-\hspace{-0.05cm}\bw^{m+1}\|\\
&\hspace{-0.06cm}-\hspace{-0.08cm}\|\mathbf{e}_x^{m}\|\|\x\hspace{-0.05cm}-\hspace{-0.05cm}\x^{m+1}\|\hspace{-0.08cm}-\hspace{-0.08cm}\|\mathbf{e}_w^{m}\|\|{\bw}\hspace{-0.05cm}-\hspace{-0.05cm}\bw^{m+1}\|\hspace{-0.08cm}-\hspace{-0.08cm}\|\mathbf{e}_z^{m}\|\|\bz\hspace{-0.05cm}-\hspace{-0.05cm}\bz^{m+1}\|,\\
\end{aligned}
\end{equation}
where $(a)$ follows from the fact $\bC\x-\bz=\bb$ and $\bA\x-\bw=\mathbf{0}$.  It is simple to verify that $\mathcal{F}$ is a compact set, given that  $\mathcal{X}$ is compact.  Furthermore, note that $(\bmu^m, \bnu^m)$ are bounded, $(\bar{\x},\bar{\bw},\bar{\bz})\in\mathcal{F}$,  $\|\be^m\|\rightarrow 0$, and $\{(\bx^{m+1},\bw^{m+1},\bz^{m+1})\}$ is bounded. Then, by letting $m\to\infty$ in \eqref{convergence3}, we get
$$\langle\nabla f(\bar{\bw})+\hspace{-0.03cm}\nabla g(\bar{\bw}),{\bw}-\bar{\bw}\rangle+\langle \nabla h(\bar{\x}),{\x}-\bar{\x}\rangle\hspace{-0.05cm}\geq 0,\hspace{-0.05cm}
$$
which gives \eqref{optcondition} and completes the proof.
\section{Proof of Proposition \ref{proposition2}}\label{appendixD}
In this appendix, we give the proof of Proposition \ref{proposition2}.  We first give an outline of the proof, which contains three main steps. \\
\textbf{Step 1: {Control the violation of the first-order optimality condition}}. For $(\x^{(r+1)},\bw^{(r+1)},\ \bz^{(r+1)})$,  we bound its violation of the first-order optimality condition by the difference between successive iterations. Specifically, based on the update rules of the proposed BSUM Algorithm \ref{alg:BSUM},  we find an $\mathbf{e}^{(r+1)}$ with 
\begin{equation*}
\label{def:e}
\begin{aligned}
\mathbf{e}^{(r+1)}\in\nabla &\mathcal{L}_m(\bx^{(r+1)},\bw^{(r+1)}\,\bz^{(r+1)})\\&\qquad\qquad+\partial \mathbb{I}_{\mathcal{X}}({\x}^{(r+1)})+\partial\mathbb{I}_{\mathcal{Z}}({\bz}^{(r+1)})
\end{aligned}
\end{equation*}
 that satisfies
\begin{equation}\label{error0}\begin{aligned}
\|\be^{(r+1)}\|^2
\leq &C_0\,(\rho^m)^2\left(\|\x^{(r+1)}-\x^{(r)}\|^2\right.\hspace{-0.1cm}\\&+\hspace{-0.1cm}\left.\|\bw^{(r+1)}-\bw^{(r)}\|^2\hspace{-0.1cm}+\hspace{-0.1cm}\|\bz^{(r+1)}-\bz^{(r)}\|^2\right),
\end{aligned}
\end{equation}
where $C_0$ is a constant independent of $m$.\\
\textbf{Step 2: Establish the sufficient decrease and lower boundedness of the objective function}. We show that the objective function is sufficiently decreased at each iteration:
\begin{equation}\label{decrease0}\begin{aligned}
&\mathcal{L}_m(\x^{(r+1)},\bw^{(r+1)},\bz^{(r+1)})-\mathcal{L}_m(\x^{(r)},\bw^{(r)},\bz^{(r)})\\
\leq\,& -C_1\rho^m\left(\|\x^{(r+1)}-\x^{(r)}\|^2-\|\bw^{(r+1)}-\bw^{(r)}\|^2\right.\\&\qquad\qquad\quad\left.-\|\bz^{(r+1)}-\bz^{(r)}\|^2\right)
\end{aligned}\end{equation}
and is bounded from below. The sufficient decrease  is established based on the fact that a strongly convex upper bound is minimized for each block at each iteration, and the boundedness is due to the specific structure of $\mathcal{L}_m$ and the boundedness of the sequence.\\
\textbf{Step 3: Identify an upper bound of $r$ with $\|\mathbf{e}^{(r)}\|\leq \epsilon_m$ by combining Steps 1 and 2.}

We next give the detailed proof and consider the above three steps separately. 

\textbf{Step 1:} From the update rules of the BSUM algorithm given by \eqref{xsolution}, \eqref{wsolution}, \eqref{zsolution},  we have
\begin{equation*}
\begin{aligned}
\mathbf{0}&\in\hspace{-0.05cm}\nabla_{\x} \mathcal{L}_m(\bx^{(r)},\hspace{-0.06cm}\bw^{(r)}, \hspace{-0.06cm}\bz^{(r)})\hspace{-0.05cm}+\hspace{-0.05cm}\gamma^m(\x^{(r+1)}\hspace{-0.06cm}-\hspace{-0.06cm}\x^{(r)})\hspace{-0.05cm}+\hspace{-0.06cm}\partial\mathbb{I}_{\mathcal{X}}(\x^{(r+1)}),\\
\mathbf{0}&=\hspace{-0.05cm}\nabla_{\bw} \mathcal{L}_m(\bx^{(r+1)},\bw^{(r+1)},\bz^{(r+1)})\hspace{-0.06cm}+\hspace{-0.06cm}\nabla g(\bw^{(r)})\hspace{-0.06cm}-\hspace{-0.05cm}\nabla g(\bw^{(r+1)}),\\
\mathbf{0}&\in\hspace{-0.05cm}\nabla_{\bz} \mathcal{L}_m(\bx^{(r+1)},\bw^{(r+1)},\bz^{(r+1)})+\partial\mathbb{I}_{\mathcal{Z}}(\bz^{(r+1)}).
\end{aligned}
\end{equation*}
Define 
\begin{equation}\label{def:e2}
\hspace{-0.3cm}\begin{aligned}
\be_x^{(r+1)}\hspace{-0.08cm}&= \hspace{-0.05cm}\nabla_{\x} \mathcal{L}_m(\bx^{(r+1)}\hspace{-0.02cm},\hspace{-0.03cm}\bw^{(r+1)}\hspace{-0.02cm}, \hspace{-0.03cm}\bz^{(r+1)}\hspace{-0.025cm})\hspace{-0.05cm}\\&\quad-\hspace{-0.05cm}\nabla_{\x} \mathcal{L}_m(\bx^{(r)},\bw^{(r)}, \bz^{(r)}\hspace{-0.025cm})\hspace{-0.05cm}-\hspace{-0.05cm}\gamma^m(\x^{(r+1)}\hspace{-0.05cm}-\hspace{-0.05cm}\x^{(r)}),\\
\be_w^{(r+1)}\hspace{-0.08cm}&=\nabla g(\bw^{(r+1)})-\nabla g(\bw^{(r)}).
\end{aligned}
\end{equation}
Then 
\begin{equation}\label{def:e}\begin{aligned}\mathbf{e}^{(r+1)}:=\left[\begin{smallmatrix}\mathbf{e}^{(r+1)}_{x}\\\mathbf{e}^{(r+1)}_{w}\\\mathbf{0}\end{smallmatrix}\right]\in\,&\nabla \mathcal{L}_m(\bx^{(r+1)},\bw^{(r+1)}\,\bz^{(r+1)})\\&+\partial \mathbb{I}_{\mathcal{X}}({\x}^{(r+1)})+\partial\mathbb{I}_{\mathcal{Z}}({\bz}^{(r+1)}).
\end{aligned}\end{equation}
From Lemma \ref{bounded}, $\{\bw^{(r+1)}\}$ lies in a bounded set, which we denote by $\mathcal{W}$. Let $L_g$ be the Lipschitz constant of $\nabla g(\bw)$ on $\mathcal{W}$. Then using the definitions of $\mathbf{e}^{(r+1)}$ and $\mathcal{L}_m$, we have
\begin{equation}\label{err_bound}
\begin{aligned}
&\|\mathbf{e}^{(r+1)}\|^2\\\leq\,& 3\|(\gamma^m+2\lambda)\mathbf{I}-\rho_\mu^m\bC^T\bC-\rho_\nu^m\bA^T\bA\|^2\|\x^{(r+1)}-\x^{(r)}\|^2\\
&+\left(3(\rho_\nu^m)^2\|\bA\|^2+L_g^2\right)\|\bw^{(r+1)}-\bw^{(r)}\|^2\\&+3(\rho_\mu^m)^2\|\bC\|^2\|\bz^{(r+1)}-\bz^{(r)}\|^2.
\end{aligned}
\end{equation}
Note that  $\rho_\mu^m=a_\mu\rho^m$,  $\rho_\nu^m=a_\nu\rho^m$, and $\gamma^m=\rho^m\|a_\mu\bC^T\bC+a_\nu\bA^T\bA\|$, where $\rho^m$ is increasing in $m$ and thus lower-bounded by $\rho^0>0$. Hence,  according to \eqref{err_bound}, we can find a  constant $C_0>0$ independent of $m$ such that \eqref{error0} holds, which completes the first step.

\textbf{Step 2:} We next establish the sufficient decrease of $\mathcal{L}_m$. Specifically, 
$$
\begin{aligned}
&\mathcal{L}_m(\x^{(r+1)},\bw^{(r)},\bz^{(r)})-\mathcal{L}_m(\x^{(r)},\bw^{(r)},\bz^{(r)})\\
\overset{(a)}{\leq}\,& u_x(\x^{(r+1)};\x^{(r)},\bw^{(r)},\bz^{(r)})-u_x(\x^{(r)}; \x^{(r)},\bw^{(r)},\bz^{(r)})\\
\overset{(b)}{\leq}\,& -\frac{\gamma^m}{2}\|\x^{(r+1)}-\x^{(r)}\|^2,
\end{aligned}$$
where $(a)$ is due to the facts that $u_x(\x;\x^{(r)},\bw^{(r)},\bz^{(r)})$ is an upper bound of $\mathcal{L}_m(\x,\bw^{(r)},\bz^{(r)})$ and $u_x(\x^{(r)};\x^{(r)},\bw^{(r)},\bz^{(r)}) =\mathcal{L}_m(\x^{(r)},\bw^{(r)},\bz^{(r)})$ and (b) follows from the update rule $\x^{(r+1)}\in\arg\min_{\x\in\mathcal{X}}~u_x(\x;\x^{(r)}, \bw^{(r)},\bz^{(r)})$ and the fact that  $u_x(\x;\x^{(r)},\bw^{(r)},\bz^{(r)})$ is strongly convex with modulus $\gamma^m$. Similarly, we can show that 
$$
\begin{aligned}
&\mathcal{L}_m(\x^{(r+1)},\bw^{(r+1)},\bz^{(r)})-\mathcal{L}_m(\x^{(r+1)},\bw^{(r)},\bz^{(r)})\\
\leq\, &-\frac{\rho_\nu^m}{2}\|\bw^{(r+1)}-\bw^{(r)}\|^2
\end{aligned}$$
 and 
 $$
\begin{aligned}
&\mathcal{L}_m(\x^{(r+1)},\bw^{(r+1)},\bz^{(r+1)})-\mathcal{L}_m(\x^{(r+1)},\bw^{(r+1)},\bz^{(r)})\\
\leq\,& -\frac{\rho_\mu^m}{2}\|\bz^{(r+1)}-\bz^{(r)}\|^2.
\end{aligned}$$
Summing the above inequalities and letting $C_1:=\min\{a_\mu,a_\nu, \|a_\mu\bC^T\bC+a_\nu\bA^T\bA\|\}/2$, we get  the  sufficient decrease of the objective function in \eqref{decrease0}. Since the sequence $\{(\x^{(r+1)},\bw^{(r+1)},\bz^{(r+1)})\}$ and the multipliers $(\bmu^m,\bnu^m)$ are bounded, the objective function $\mathcal{L}_m$ is bounded from below.  We denote the lower bound of $\mathcal{L}_m$ by $\underline{{L}}$. 

\textbf{Step 3}: Let 
\begin{equation}\label{repsilonm}
r(\epsilon_m):=\min\{r\geq0\mid\|\be^{(r)}\|\leq \epsilon_m\}.
\end{equation} Our goal is to determine an upper bound of $r(\epsilon_m)$. Given any $R\geq1$, combining \eqref{error0} and \eqref{decrease0} and summing them from $r=0$ to $r=R-1$, we have 
\begin{equation}\label{error_sum}
\begin{aligned}
\sum_{r=1}^{R}\|\be^{(r)}\|^2\leq \frac{C_0\rho^m}{C_1}(&\mathcal{L}_m(\x^{(0)},\bw^{(0)},\bz^{(0)})\\&-\mathcal{L}_m(\x^{(R)},\bw^{(R)},\bz^{(R)})).
\end{aligned}
\end{equation}
Recalling the definition of $\mathcal{L}_m$ in \eqref{def:lm}, we can upper bound $\mathcal{L}_m(\x^{(0)},\bw^{(0)},\bz^{(0)})$ as 
$$\mathcal{L}_m(\x^{(0)},\bw^{(0)},\bz^{(0)})\leq C_2{\rho}^m,$$
where $C_2$ is a positive constant independent of $m$. Then, letting $R=r(\epsilon_m)-1$ in \eqref{error_sum} and noting the definition of $r(\epsilon_m)$, we have 
$$(r(\epsilon_m)-1)\epsilon_m^2\leq \sum_{r=1}^{r(\epsilon_m)-1}\|\be^{(r)}\|^2\leq \frac{C_0C_2}{C_1}(\rho^m)^2-\underline{L},$$
which implies that $r(\epsilon_m)=\mathcal{O}\left(\frac{(\rho^m)^2}{\epsilon_m^2}\right).$

\bibliographystyle{IEEEtran}
\bibliography{IEEEabrv,ISAC_QCE}

\begin{thebibliography}{10}
\providecommand{\url}[1]{#1}
\csname url@samestyle\endcsname
\providecommand{\newblock}{\relax}
\providecommand{\bibinfo}[2]{#2}
\providecommand{\BIBentrySTDinterwordspacing}{\spaceskip=0pt\relax}
\providecommand{\BIBentryALTinterwordstretchfactor}{4}
\providecommand{\BIBentryALTinterwordspacing}{\spaceskip=\fontdimen2\font plus
\BIBentryALTinterwordstretchfactor\fontdimen3\font minus
  \fontdimen4\font\relax}
\providecommand{\BIBforeignlanguage}[2]{{%
\expandafter\ifx\csname l@#1\endcsname\relax
\typeout{** WARNING: IEEEtran.bst: No hyphenation pattern has been}%
\typeout{** loaded for the language `#1'. Using the pattern for}%
\typeout{** the default language instead.}%
\else
\language=\csname l@#1\endcsname
\fi
#2}}
\providecommand{\BIBdecl}{\relax}
\BIBdecl

\bibitem{CRSS}
L.~Zheng, M.~Lops, Y.~C. Eldar, and X.~Wang, ``Radar and communication
  coexistence: An overview,'' \emph{IEEE Signal Process. Mag.}, vol.~36, no.~5,
  pp. 85--99, Sept. 2019.

\bibitem{coexist}
F.~Liu, C.~Masouros, A.~Li, T.~Ratnarajah, and J.~Zhou, ``{MIMO} radar and
  cellular coexistence: A power-efficient approach enabled by interference
  exploitation,'' \emph{IEEE Trans. Signal Process.}, vol.~66, no.~14, pp.
  3681--3695, Jul. 2018.

\bibitem{DFRCapp}
P.~Kumari, J.~Choi, N.~González-Prelcic, and R.~W. Heath, ``{IEEE
  802.11}ad-based radar: An approach to joint vehicular communication-radar
  system,'' \emph{IEEE Trans. Veh. Technol.}, vol.~67, no.~4, pp. 3012--3027,
  Apr. 2018.

\bibitem{DFRC2}
F.~Liu, L.~Zhou, C.~Masouros, A.~Li, W.~Luo, and A.~Petropulu, ``Toward
  dual-functional radar-communication systems: Optimal waveform design,''
  \emph{IEEE Trans. Signal Process.}, vol.~66, no.~16, pp. 4264--4279, Aug.
  2018.

\bibitem{DFRC3}
X.~Liu, T.~Huang, N.~Shlezinger, Y.~Liu, J.~Zhou, and Y.~C. Eldar, ``Joint
  transmit beamforming for multiuser {MIMO} communications and {MIMO} radar,''
  \emph{IEEE Trans. Signal Process.}, vol.~68, pp. 3929--3944, 2020.

\bibitem{DFRC4}
F.~Liu, Y.-F. Liu, A.~Li, C.~Masouros, and Y.~C. Eldar, ``Cramér-{Rao} bound
  optimization for joint radar-communication beamforming,'' \emph{IEEE Trans.
  Signal Process.}, vol.~70, pp. 240--253, 2022.

\bibitem{survey_DFRC}
F.~Liu, Y.~Cui, C.~Masouros, J.~Xu, T.~X. Han, Y.~C. Eldar, and S.~Buzzi,
  ``Integrated sensing and communications: Toward dual-functional wireless
  networks for {6G} and beyond,'' \emph{IEEE J. Sel. Areas Commun.}, vol.~40,
  no.~6, pp. 1728--1767, Jun. 2022.

\bibitem{survey_DFRC2}
J.~A. Zhang, F.~Liu, C.~Masouros, R.~W. Heath, Z.~Feng, L.~Zheng, and
  A.~Petropulu, ``An overview of signal processing techniques for joint
  communication and radar sensing,'' \emph{IEEE J. Sel. Top. Signal Process.},
  vol.~15, no.~6, pp. 1295--1315, Nov. 2021.

\bibitem{wang2023globally}
\BIBentryALTinterwordspacing
Z.~Wang, J.~Wu, Y.-F. Liu, and F.~Liu, ``Globally optimal beamforming design
  for integrated sensing and communication systems,'' 2023. [Online].
  Available: \url{http://arxiv.org/abs/2309.06674}
\BIBentrySTDinterwordspacing

\bibitem{ITU}
I.-R. WP5D, ``Draft new recommendation {ITU-RM.[IMT.FRAMEWORK FOR 2030 AND
  BEYOND]},'' 2023.

\bibitem{massivemimo2}
F.~Rusek, D.~Persson, B.~K. Lau, E.~G. Larsson, T.~L. Marzetta, O.~Edfors, and
  F.~Tufvesson, ``Scaling up {MIMO}: Opportunities and challenges with very
  large arrays,'' \emph{IEEE Signal Process. Mag.}, vol.~30, no.~1, pp. 40--60,
  Jan. 2013.

\bibitem{PAbook}
S.~C. Cripps, \emph{RF Power Amplifiers for Wireless Communications},
  2nd~ed.\hskip 1em plus 0.5em minus 0.4em\relax Norwood, MA, USA: Artech
  House, 2006.

\bibitem{SQUID}
S.~Jacobsson, G.~Durisi, M.~Coldrey, T.~Goldstein, and C.~Studer, ``Quantized
  precoding for massive {MU-MIMO},'' \emph{IEEE Trans. Commun.}, vol.~65,
  no.~11, pp. 4670--4684, Nov. 2017.

\bibitem{diversity}
Z.~Wu, J.~Wu, W.-K. Chen, and Y.-F. Liu, ``Diversity order analysis for
  quantized constant envelope transmission,'' \emph{IEEE Open J. Signal
  Process.}, vol.~4, pp. 21--30, 2023.

\bibitem{analysis2}
Z.~Wu, J.~Ma, Y.-F. Liu, and A.~L.~Swindlehurst, ``Asymptotic {SEP} analysis
  and optimization of linear-quantized precoding in massive {MIMO} systems,''
  \emph{IEEE Trans. Inf. Theory (early access)}, 2023.

\bibitem{trellis}
M.~Kazemi, H.~Aghaeinia, and T.~M. Duman, ``Discrete-phase constant envelope
  precoding for massive {MIMO} systems,'' \emph{IEEE Trans. Commun.}, vol.~65,
  no.~5, pp. 2011--2021, May 2017.

\bibitem{GEMM}
M.~Shao, Q.~Li, W.-K. Ma, and A.~M.-C. So, ``A framework for one-bit and
  constant-envelope precoding over multiuser massive {MISO} channels,''
  \emph{IEEE Trans. Signal Process.}, vol.~67, no.~20, pp. 5309--5324, Oct.
  2019.

\bibitem{ciqce}
H.~Jedda, A.~Mezghani, A.~L. Swindlehurst, and J.~A. Nossek, ``Quantized
  constant envelope precoding with {PSK} and {QAM} signaling,'' \emph{IEEE
  Trans. Wireless Commun.}, vol.~17, no.~12, pp. 8022--8034, Dec. 2018.

\bibitem{ICASSP2023}
Z.~Wu, Y.-F. Liu, B.~Jiang, and Y.-H. Dai, ``Efficient quantized constant
  envelope precoding for multiuser downlink massive {MIMO} systems,'' in
  \emph{Proc. IEEE Int. Conf. Acoust., Speech, Signal Process.}, Rhodes Island,
  Greece, Jun. 2023, pp. 1--5.

\bibitem{CI1}
C.~Masouros, ``Correlation rotation linear precoding for {MIMO} broadcast
  communications,'' \emph{IEEE Trans. Signal Process.}, vol.~59, no.~1, pp.
  252--262, Jan. 2011.

\bibitem{CI2}
C.~Masouros, M.~Sellathurai, and T.~Ratnarajah, ``Vector perturbation based on
  symbol scaling for limited feedback {MISO} downlinks,'' \emph{IEEE Trans.
  Signal Process.}, vol.~62, no.~3, pp. 562--571, Feb. 2014.

\bibitem{CI3}
C.~Masouros and G.~Zheng, ``Exploiting known interference as green signal power
  for downlink beamforming optimization,'' \emph{IEEE Trans. Signal Process.},
  vol.~63, no.~14, pp. 3628--3640, Jul. 2015.

\bibitem{CItutorial}
A.~Li, D.~Spano, J.~Krivochiza, S.~Domouchtsidis, C.~G. Tsinos, C.~Masouros,
  S.~Chatzinotas, Y.~Li, B.~Vucetic, and B.~Ottersten, ``A tutorial on
  interference exploitation via symbol-level precoding: Overview,
  state-of-the-art and future directions,'' \emph{IEEE Commun. Surveys Tuts.},
  vol.~22, no.~2, pp. 796--839, 2nd Quart. 2020.

\bibitem{radar_QCE_3}
S.~Ahmed, J.~S. Thompson, Y.~R. Petillot, and B.~Mulgrew, ``Finite alphabet
  constant-envelope waveform design for {MIMO} radar,'' \emph{IEEE Trans.
  Signal Process.}, vol.~59, no.~11, pp. 5326--5337, Nov. 2011.

\bibitem{radar_QCE_1}
A.~Mezghani and R.~W. Heath, ``{MIMO} beampattern and waveform design with low
  resolution {DACs},'' in \emph{Proc. IEEE Radar Conf.}, Boston, MA, USA, Apr.
  2019, pp. 1--6.

\bibitem{radar_QCE_2}
J.~Wen and B.~Liao, ``A block coordinate descent approach to design constant
  modulus waveform for {MIMO} radar with low-resolution {DACs},'' in
  \emph{Proc. IEEE Int. Conf. Inf. Commun. Signal Process.}, Shenzhen, China,
  Nov. 2022, pp. 349--354.

\bibitem{analysis1}
A.~K. Saxena, I.~Fijalkow, and A.~L. Swindlehurst, ``Analysis of one-bit
  quantized precoding for the multiuser massive {MIMO} downlink,'' \emph{IEEE
  Trans. Signal Process.}, vol.~65, no.~17, pp. 4624--4634, Sept. 2017.

\bibitem{CIfirst}
H.~Jedda, A.~Mezghani, J.~A. Nossek, and A.~L. Swindlehurst, ``Massive {MIMO}
  downlink 1-bit precoding with linear programming for {PSK} signaling,'' in
  \emph{Proc. IEEE Workshop Signal Process. Adv. Wireless Commun.}, Jul. 2017,
  pp. 1--5.

\bibitem{CImodel}
A.~Li, C.~Masouros, F.~Liu, and A.~L. Swindlehurst, ``Massive {MIMO} 1-bit
  {DAC} transmission: A low-complexity symbol scaling approach,'' \emph{IEEE
  Trans. Wireless Commun.}, vol.~17, no.~11, pp. 7559--7575, Nov. 2018.

\bibitem{sep2}
F.~Sohrabi, Y.-F. Liu, and W.~Yu, ``One-bit precoding and constellation range
  design for massive {MIMO} with {QAM} signaling,'' \emph{IEEE J. Sel. Topics
  Signal Process.}, vol.~12, no.~3, pp. 557--570, Jun. 2018.

\bibitem{NL1P}
Z.~Wu, B.~Jiang, Y.-F. Liu, M.~Shao, and Y.-H. Dai, ``Efficient {CI}-based
  one-bit precoding for multiuser downlink massive {MIMO} systems with {PSK}
  modulation,'' \emph{IEEE Trans. Wireless Commun. (early access)}, 2023.

\bibitem{radar_1bit1}
M.~Deng, Z.~Cheng, L.~Wu, B.~Shankar, and Z.~He, ``One-bit {ADCs/DACs} based
  {MIMO} radar: Performance analysis and joint design,'' \emph{IEEE Trans.
  Signal Process.}, vol.~70, pp. 2609--2624, 2022.

\bibitem{radar_1bit2}
Z.~Cheng, B.~Liao, Z.~He, and J.~Li, ``Transmit signal design for large-scale
  {MIMO} system with 1-bit {DACs},'' \emph{IEEE Trans. Wireless Commun.},
  vol.~18, no.~9, pp. 4466--4478, Sept. 2019.

\bibitem{radar_1bit3}
M.~Deng, Z.~Cheng, and Z.~He, ``Spectrally compatible waveform design for
  large-scale {MIMO} radar beampattern synthesis with one-bit {DACs},''
  \emph{IEEE Trans. Aerosp. Electron. Syst.}, vol.~58, no.~5, pp. 4729--4744,
  Oct. 2022.

\bibitem{DFRC1bit3}
Z.~Cheng, S.~Shi, Z.~He, and B.~Liao, ``Transmit sequence design for
  dual-function radar-communication system with one-bit {DACs},'' \emph{IEEE
  Trans. Wireless Commun.}, vol.~20, no.~9, pp. 5846--5860, Sept. 2021.

\bibitem{DFRC1bit2}
J.~Yan and J.~Zheng, ``Transmit signal design of {MIMO} dual-function radar
  communication with 1-bit {DACs},'' in \emph{Proc. IEEE Workshop Signal
  Process. Adv. Wireless Commun.}, Jul. 2022, pp. 1--5.

\bibitem{DFRC1bit1}
X.~Yu, Q.~Yang, Z.~Xiao, H.~Chen, V.~Havyarimana, and Z.~Han, ``A precoding
  approach for dual-functional radar-communication system with one-bit
  {DACs},'' \emph{IEEE J. Sel. Areas Commun.}, vol.~40, no.~6, pp. 1965--1977,
  Jun. 2022.

\bibitem{DFRC_CI1}
R.~Liu, M.~Li, Q.~Liu, and A.~L. Swindlehurst, ``Dual-functional
  radar-communication waveform design: A symbol-level precoding approach,''
  \emph{IEEE J. Sel. Topics Signal Process.}, vol.~15, no.~6, pp. 1316--1331,
  Nov. 2021.

\bibitem{DFRC_CI2}
Z.~Zhang, Q.~Chang, F.~Liu, and S.~Yang, ``Dual-functional radar-communication
  waveform design: Interference reduction versus exploitation,'' \emph{IEEE
  Commun. Lett.}, vol.~26, no.~1, pp. 148--152, Jan. 2022.

\bibitem{DFRC_CI3}
N.~Su, F.~Liu, Z.~Wei, Y.-F. Liu, and C.~Masouros, ``Secure dual-functional
  radar-communication transmission: Exploiting interference for resilience
  against target eavesdropping,'' \emph{IEEE Trans. Wireless Commun.}, vol.~21,
  no.~9, pp. 7238--7252, Sept. 2022.

\bibitem{radaronly}
P.~Stoica, J.~Li, and Y.~Xie, ``On probing signal design for {MIMO} radar,''
  \emph{IEEE Trans. Signal Process.}, vol.~55, no.~8, pp. 4151--4161, Aug.
  2007.

\bibitem{JSACsurvey}
\BIBentryALTinterwordspacing
Y.-F. Liu, T.-H. Chang, M.~Hong, Z.~Wu, A.~M.-C. So, E.~A. Jorswieck, and
  W.~Yu, ``A survey of advances in optimization methods for wireless
  communication system design,'' 2024. [Online]. Available:
  \url{http://arxiv.org/abs/2401.12025}
\BIBentrySTDinterwordspacing

\bibitem{nonlinear}
D.~P. Bertsekas, \emph{Nonlinear Programming}.\hskip 1em plus 0.5em minus
  0.4em\relax Belmont, MA, USA: Athena Scientific, 1999.

\bibitem{boundALM}
G.~Galvan, M.~Lapucci, T.~Levato, and M.~Sciandrone, ``An alternating augmented
  {L}agrangian method for constrained nonconvex optimization,'' \emph{Optim.
  Methods Softw.}, vol.~35, no.~3, pp. 502--520, 2020.

\bibitem{BSUM}
M.~Razaviyayn, M.~Hong, and Z.-Q. Luo, ``A unified convergence analysis of
  block successive minimization methods for nonsmooth optimization,''
  \emph{SIAM J. Optim.}, vol.~23, no.~2, pp. 1126--1153, 2013.

\end{thebibliography}

\end{document}